\documentclass[11pt]{article}
\usepackage[margin=1in]{geometry}
\usepackage[T1]{fontenc}
\usepackage{lmodern}
\usepackage{amsmath,amssymb,amsthm,mathtools}
\usepackage{microtype}
\usepackage{enumitem}
\usepackage{array}
\usepackage{xcolor}
\definecolor{green}{RGB}{0,128,0}
\usepackage[hypertexnames=false,colorlinks=true,linkcolor=blue!45!black,citecolor=blue!45!black,urlcolor=blue!45!black]{hyperref}
\setlist{itemsep=3pt,topsep=5pt,leftmargin=*}
\newtheorem{theorem}{Theorem}[section]
\newtheorem{lemma}[theorem]{Lemma}
\newtheorem{proposition}[theorem]{Proposition}
\newtheorem{corollary}[theorem]{Corollary}
\theoremstyle{definition}
\newtheorem{definition}[theorem]{Definition}
\newtheorem{remark}[theorem]{Remark}
\newcommand{\Fp}{\mathbb F_p}
\newcommand{\bits}{\{0,1\}}
\newcommand{\Span}{\operatorname{span}}
\newcommand{\rank}{\operatorname{rank}}
\newcommand{\wt}{\operatorname{wt}}
\newcommand{\poly}{\operatorname{poly}}
\newcommand{\Adv}{\operatorname{Adv}}
\newcommand{\Test}{\mathcal T}
\newcommand{\Smat}{\mathcal S}
\newcommand{\Vmat}{\mathcal V}
\DeclareMathOperator{\Enc}{Enc}
\DeclareMathOperator{\Dec}{Dec}
\title{Witness Encryption via Prime-Order Generic Groups}
\author{Isaac M Hair \thanks{\texttt{isaacmhair@gmail.com}} \\ UCSB, UCLA \and Amit Sahai \thanks{\texttt{sahai@cs.ucla.edu}} \\ UCLA}
\date{}

\begin{document}
\maketitle
\thispagestyle{empty}
\vspace{-2em}
\begin{abstract}
For every constant $c > 0$, we unconditionally construct witness encryption for NP in the classical generic cyclic group model with the following properties. For SAT instances of size $n$, encryption runs in time $\poly(n)$, and any satisfying assignment can be used to decrypt in $\poly(n)$ time with correctness error $2^{-n}$. If no satisfying assignment exists, every generic adversary making at most $n^{(\log n)^c}$ group queries has distinguishing advantage at most $n^{-(\log n)^c}$.

Along the way, we give the first superconstant-factor NP-hardness of approximation result for homogeneous MinRank under randomized polynomial-time reductions. We obtain a gap of $(\log n)^c$ for any prescribed constant $c>0$, even when YES instances are required to have a rank-one matrix of the form $\mathbf u\mathbf v^T$ where $\mathbf v$ is a 0/1 vector.

\emph{The main results of this paper were discovered by Codex / ChatGPT 6 Astra Ultra using harness components from the UCLA Moonshot Harness Project, in response to a general strategy for witness encryption suggested by the authors. See the end of the introduction for a discussion.}
\end{abstract}

\newpage
\setcounter{page}{1}

\section{Introduction}
Witness encryption (WE) is an extension of public key encryption introduced by Garg, Gentry, Sahai, and
Waters~\cite{GargGentrySahaiWaters2013}. It allows a message to be encrypted with respect to any mathematical
statement: anyone who knows a witness that the statement is true can
decrypt, but if the statement is false, the ciphertext must hide the message. Several candidate WE schemes have been proposed~\cite{GargGentrySahaiWaters2013,GentryLewkoWaters2014,ChenVaikuntanathanWee2018,BartusekEtAl2020,Tsabary2022,VaikuntanathanWeeWichs2022,SoukhanovEtAl2026}, but so far the only ones based on standard assumptions route through indistinguishability obfuscation~\cite{GargEtAl2013Obfuscation,JainLinSahai2021,JainLinSahai2022,RagavanVafaVaikuntanathan2024}.

None of the existing witness encryption schemes are known to be unconditionally secure in the prime-order generic-group model, or in any other related algebraic ideal model that is known to have plausibly secure algebraic realizations.
In the prime-order generic group model, group elements have random encodings, and algorithms access
the group law through an oracle. There is no pairing or multilinear
operation. Barta, Ishai, Ostrovsky, and Wu~\cite{BartaIshaiOstrovskyWu2020}
identified a route to witness encryption in this model based on an (as-of-yet unproven) conjecture on NP hardness of approximation for the minimum distance problem. Their approach raises a
useful broader question: which NP hardness of approximation gaps lead to a witness encryption scheme that can simultaneously support both
efficient decryption and a generic security proof?

Our answer uses matrix rank. The reduction produces a linear space of
matrices with two properties. On a true statement, a witness identifies
a rank-one matrix \(\mathbf u\mathbf v^T\), where \(\mathbf v\) has small integer coordinates.
On a false statement, every nonzero matrix in the space has large rank.
The first property allows us to perform decryption in polynomial time, and the second makes the algebraic equalities that a generic adversary can test unlikely.

\subsection{Our Results}
We require correctness on true statements and security on false statements; a true statement need not hide the message from a party lacking a witness (Definition~\ref{def:we}). Throughout, logarithms have base two. Our main result is as follows:

\begin{theorem}[Generic Group Witness Encryption]\label{thm:main}
For every constant $c>0$, there are polynomial-time
witness-encryption algorithms for every NP relation in the classical prime-order generic-group model with the
following guarantees. Let \(n=\max\{2,\lambda,|\varphi|\}\), where
\(\lambda\) is the security parameter and \(\varphi\) is the encoded
NP statement. Then:
\begin{enumerate}
\item Every valid witness decrypts correctly with probability at least
\(1-2^{-n}\).
\item On every false statement, every classical adversary making at most
\(n^{(\log n)^c}\) group-oracle queries has distinguishing advantage at most
\(n^{-(\log n)^c}\).
\end{enumerate}
\end{theorem}

Our construction gives the first superconstant-factor NP hardness of approximation result for homogeneous MinRank, which asks for the minimum rank of a nonzero matrix in a linear space $\Smat\subseteq\Fp^{m\times m}$ given by a basis. We obtain a stronger bounded-factor promise: in the YES case, $\Smat$ contains a nonzero rank-one matrix $\mathbf u\mathbf v^T$ with $\mathbf v\in\bits^m$, but in the NO case, every nonzero matrix has large rank.

\begin{theorem}[Bounded-Factor MinRank Hardness]
\label{thm:minrank-informal}
For every fixed constant $c>0$, it is NP-hard under randomized polynomial-time reductions to distinguish the following cases for a linear space $\Smat\subseteq\Fp^{m\times m}$, given by a basis,
with $m\ge4$ and primes $p=\Theta(2^m)$:
\begin{enumerate}
\item $\Smat$ contains a nonzero matrix $\mathbf u\mathbf v^T$ with $\mathbf v\in\bits^m$.
\item Every nonzero matrix in $\Smat$ has rank at least $(\log m)^c$.
\end{enumerate}
\end{theorem}

See Theorem~\ref{thm:minrank} for the reduction, which is deterministic
for every supplied prime $p\ge2^m$, and Section~\ref{sec:together}
for the randomized choice of prime.  Previous polynomial-time reductions gave only
constant-factor NP-hardness~\cite{GuruswamiRenTang2026}; obtaining a larger gap with their techniques requires a superpolynomial time reduction. In Theorem~\ref{thm:we-from-minrank}, we give a WE scheme based on hardness of bounded-factor MinRank, and we combine this with Theorem~\ref{thm:minrank} in Section~\ref{sec:together} to prove Theorem~\ref{thm:main}.

Our bounded-factor MinRank reduction builds on the moment-matrix approach of Guruswami, Ren, and Tang~\cite{GuruswamiRenTang2026}. We incorporate ideas from the literature on subspace designs and symbolic rank condensers~\cite{ForbesShpilka2012,GuruswamiKopparty2016,ForbesSaptharishiShpilka2014,ForbesGuruswami2015}, polynomial identity testing~\cite{RazShpilka2005}, superposition hardness for quadratic equations~\cite{KhotSaket2017}, Boolean optimization hierarchies~\cite{LovaszSchrijver1991Cones,Laurent2003}, Nullstellensatz and polynomial calculus proof systems~\cite{BeameImpagliazzoKrajicekPitassiPudlak1996,Buss1998Designs,CleggEdmondsImpagliazzo1996}, and finite-rank moment reconstruction~\cite{LaurentMourrain2009,Mourrain2018}. Section~\ref{sec:overview} gives a technical overview of our rank-gap construction, starting with a $1$-versus-$3$ gap. Section~\ref{sec:minrank} then gives a warm-up version of Theorem~\ref{thm:minrank} with gap $\Omega(\log m)$, and Section~\ref{sec:polylogminrank} extends it to the stated polylogarithmic gap. We discuss these connections alongside the corresponding constructions.

\subsection{Cryptographic Implications}\label{sec:applications}
We combine Theorem~\ref{thm:main} with known constructions to obtain three applications secure against probabilistic polynomial-time classical generic adversaries, and a standard-model consequence under an additional assumption about a concrete group instantiation. All schemes have negligible correctness error; proofs appear in Section~\ref{sec:together}.

Garg, Gentry, Sahai, and Waters~\cite{GargGentrySahaiWaters2013} construct public-key encryption from witness encryption and one-way functions.

\begin{corollary}[Public-Key Encryption with Small Keys]\label{cor:pke}
Assuming one-way functions, there is an IND-CPA-secure public-key
encryption scheme in the classical prime-order generic-group model.
For security parameter $\lambda$, key generation samples a
$\lambda$-bit seed and evaluates a pseudorandom generator once to obtain
a $2\lambda$-bit public key. Key generation uses no group queries.
\end{corollary}

Waters and Wichs~\cite{WatersWichs2024} construct attribute-based encryption (ABE) from witness encryption, one-way functions, and statistically sound noninteractive zero-knowledge proofs. LWE supplies the auxiliary primitives.

\begin{corollary}[Attribute-Based Encryption]\label{cor:abe}
Assuming LWE, there is an adaptively secure key-policy attribute-based encryption
scheme for polynomial-size circuits in the classical prime-order
generic-group model.
\end{corollary}

Freitag, Waters, and Wu~\cite{FreitagWatersWu2023} use function-binding hash functions to construct flexible broadcast encryption and registered ABE from witness encryption and LWE. Users generate their own keys; flexible broadcast encryption encrypts to a chosen set of public keys, while registered ABE uses a transparent key curator in place of a master secret-key authority.

\begin{corollary}[Flexible Broadcast Encryption and Registered ABE]
\label{cor:trustless}
Assuming LWE, the following schemes exist in the classical prime-order
generic-group model:
\begin{enumerate}
\item A semi-statically secure flexible broadcast encryption scheme.
\item A policy-selectively secure registered ABE scheme for
polynomial-size circuit policies.
\end{enumerate}
Both guarantees use the security definitions of
Freitag, Waters, and Wu~\cite{FreitagWatersWu2023}, without corruption
of honestly generated keys.
\end{corollary}

Here semi-static security allows the adversary to choose the challenge recipients after seeing honestly generated public keys, while policy-selective security requires it to choose the challenge policy before setup~\cite[Definitions~4.2 and~6.3]{FreitagWatersWu2023}.

A further line of work obtains null indistinguishability obfuscation
from witness encryption and
LWE~\cite{GoyalKoppulaWaters2017,WichsZirdelis2017}.
Such an obfuscator preserves a circuit's functionality and hides which
circuit was obfuscated when both circuits always output $\bot$.
Nassar, Waters, and Wu~\cite{NassarWatersWu2026} use this route to
obtain indexed positional witness encryption and further applications.
(Indexed positional witness encryption associates each witness with an
index and permits decryption only above a hidden cutoff.)

These constructions use witness decryption as an ordinary circuit, which a concrete group implementation supplies. Our algebraic correctness bound continues to hold; we additionally assume standard-model false-statement security against classical nonuniform polynomial-time adversaries, beyond the generic guarantee of Theorem~\ref{thm:main}.

\begin{corollary}[Applications of a Concrete Instantiation]
\label{cor:concrete-applications}
Suppose that our witness-encryption construction admits an efficiently
computable prime-order-group instantiation satisfying standard-model
false-statement security against nonuniform polynomial-time adversaries.
Assume also that LWE with a subexponential modulus-to-noise ratio
is hard for such adversaries. Then the following schemes
exist in the standard model:
\begin{enumerate}
\item For every fixed $\varepsilon\in(0,1)$, indexed positional witness
encryption for NP, with ciphertext length
$N^\varepsilon\poly(\lambda,|C|,|\mu|,\log N)$, where $C$ is the
relation circuit, $\mu$ is the message, and $N$ is the size of the
index space.
\item For every fixed $\varepsilon\in(0,1)$, a collusion-resistant
broadcast-and-trace scheme with public tracing and ciphertext length
$N^\varepsilon\poly(\lambda,|\mu|,\log N)$ for $N$ users. The
recipient set is supplied separately to decryption and is not included
in this length bound.
\item Somewhere-statistically sound noninteractive batch arguments
for NP.
\end{enumerate}
Here $N$ is polynomially bounded in $\lambda$, and the polynomial
factors may depend on $\varepsilon$. Computational security is against
classical nonuniform probabilistic polynomial-time adversaries. Correctness
errors, and the statistical soundness error in the third item, are
negligible.
\end{corollary}

\subsection{Witness Encryption and Ideal Groups}
Garg, Gentry, Sahai, and Waters gave the first witness-encryption candidate from multilinear maps~\cite{GargGentrySahaiWaters2013}. Gentry, Lewko, and Waters introduced positional witness encryption and constructions from instance-independent multilinear-map assumptions~\cite{GentryLewkoWaters2014}. Zeroizing and annihilation attacks broke several concrete encodings and obfuscation instantiations~\cite{CheonEtAl2015,MilesSahaiZhandry2016}, but did not give a general attack on the witness-encryption approach~\cite{VaikuntanathanWeeWichs2022}.

Chen, Vaikuntanathan, and Wee proposed a candidate using general read-once branching programs~\cite{ChenVaikuntanathanWee2018} and lattice encodings~\cite{GentryGorbunovHalevi2015}. Tsabary gave a lattice-based candidate under a new trapdoor-sampling assumption~\cite{Tsabary2022}. In concurrent work, Vaikuntanathan, Wee, and Wichs analyzed the branching-program candidate under private-coin evasive LWE and subexponential LWE~\cite{VaikuntanathanWeeWichs2022}. Subsequent counterexamples refute the proposed private-coin evasive-LWE assumptions~\cite{BrzuskaUnalWoo2024,DottlingEtAl2025}, but do not themselves break the witness-encryption constructions.

Bartusek, Ishai, Jain, Ma, Sahai, and Zhandry proposed obfuscation and witness-encryption candidates from affine determinant programs~\cite{BartusekEtAl2020}. Soukhanov, Rebenko, El Gebali, Komarov, Kilinc-Alper,  Abdalla, and Towa extended this approach to arithmetic programs to reduce the cost of representing SNARK verification~\cite{SoukhanovEtAl2026}. El Gebali, Rebenko, Schofnegger, and Soukhanov~\cite{ElGebaliRebenkoSchofneggerSoukhanov2026} recently broke both candidates for sparse circuits using low-rank commutators of public ciphertext matrices; Jo~\cite{Jo2026} also gives a polynomial-time attack on the original candidate for an explicit family of false SUBSET-SUM instances.

The closest generic-group precursor is Barta, Ishai, Ostrovsky, and Wu~\cite{BartaIshaiOstrovskyWu2020}, who connect witness-preserving approximation hardness over large fields to witness encryption. Their proposed condition requires a superlogarithmic NP-hardness gap for Hamming minimum distance, which remains open.

Indistinguishability obfuscation implies witness encryption by obfuscating the circuit that outputs the message on valid witnesses and $\bot$ otherwise~\cite{GargEtAl2013Obfuscation}. On a false statement, both circuits always output $\bot$. The constructions of Jain, Lin, and Sahai~\cite{JainLinSahai2021,JainLinSahai2022} thus give witness encryption from subexponential versions of well-studied assumptions. Ragavan, Vafa, and Vaikuntanathan~\cite{RagavanVafaVaikuntanathan2024} subsequently obtain indistinguishability obfuscation from subexponential versions of the decisional linear assumption on bilinear groups, large-field LPN, and binary sparse LPN.

Other work relates witness encryption to proof systems.
Faonio, Nielsen, and Venturi study predictable arguments of
knowledge~\cite{FaonioNielsenVenturi2017}. Liu, Mazor, and Pass
characterize witness encryption through special-honest-verifier
zero-knowledge arguments with logarithmic prover
communication~\cite{LiuMazorPass2025}. Garg, Hajiabadi, Kolonelos,
Kothapalli, and Policharla give a framework based on linearly verifiable
SNARKs, with special-purpose applications~\cite{GargEtAl2025Framework}.
Choi and Vaudenay study extractable witness encryption for a promise
version of multi-subset sum in a hidden-group-with-hashing
abstraction~\cite{ChoiVaudenay2022}; both the promise and the model
differ from ours.

Our security proof follows the random-encoding approach to generic-group
lower bounds exemplified by Shoup~\cite{Shoup1997}. The generic-group identity-based encryption lower bounds of Schul-Ganz and Segev~\cite{SchulGanzSegev2021} and Zhandry~\cite{Zhandry2022Labels} do not by themselves exclude witness encryption in the model considered here. As discussed by Barta,
Ishai, Ostrovsky, and Wu~\cite{BartaIshaiOstrovskyWu2020}, the usual transformation from
witness encryption to identity-based encryption uses the group in a
non-black-box way.

\subsection{MinRank and Coding-Theoretic Hardness}
The NP-hardness of homogeneous rank-one MinRank follows from the
standard linearization of Boolean quadratic equations: represent an
assignment by its matrix of pairwise products and impose the equations
as linear constraints. We recall this reduction in
Section~\ref{sec:overview}. Guruswami, Ren, and
Tang~\cite{GuruswamiRenTang2026} obtain larger rank gaps.
For a source instance of size \(M\) and rank threshold \(k\), their constructions have
size \(M^{O(\log k)}\) over fixed characteristic-two fields and
\(M^{O(k)}\) over arbitrary fixed finite fields. Our construction modifies
their moment-matrix approach by replacing all low-degree monomial
weights with a much smaller family of weights.

MinRank has several cryptographic formulations. Courtois uses an affine
MinRank problem for authentication~\cite{Courtois2001}; Gaborit and
Z\'emor study decoding and minimum distance for extension-field-linear
rank codes~\cite{GaboritZemor2016}. Chatterjee, Mu, and Vasudevan
construct public-key encryption from planted average-case binary
MinRank, using matrix-valued inner products and rank-metric
duality~\cite{ChatterjeeMuVasudevan2026}. Our construction uses a
worst-case reduction to a homogeneous prime-field matrix space.
Decryption tests a short interval of scalar multiples, and security
uses the rank of every nonzero matrix in the space.

For the Hamming minimum distance problem underlying the approach of Barta, Ishai, Ostrovsky, and Wu~\cite{BartaIshaiOstrovskyWu2020},
Vardy~\cite{Vardy1997} first proved NP-hardness of computing the minimum
distance exactly. Dumer, Micciancio, and Sudan~\cite{DumerMicciancioSudan2003}
then obtained hardness within every constant factor over each fixed
finite field under randomized
polynomial-time reductions. Cheng and Wan~\cite{ChengWan2009} derandomized
this result over each fixed finite field, with subsequent simplifications and refinements by Austrin and
Khot~\cite{AustrinKhot2014}, Micciancio~\cite{Micciancio2014}, and Bhattiprolu,
Guruswami, and Ren~\cite{BhattiproluGuruswamiRen2025}.
These results give only constant gaps under polynomial-time reductions.

\section*{Acknowledgements and AI Use Methodology}
\begin{flushleft}
\renewcommand{\arraystretch}{1.3}
\begin{tabular}{@{}l@{\qquad}c@{}}
\textbf{Activity} & \textbf{AI used} \\
Asking the research question & [No] \\
Coming up with the approach & [Yes] \\
Proof development & [Yes] \\
Writing and exposition & [Yes] \\
Checking for bugs & [Yes] \\
Other supporting tasks & [Yes] \\
\end{tabular}
\end{flushleft}

The human authors spent several months thinking about how to construct witness encryption in the generic group model making use of various NP reductions. We initially attempted to get stronger NP hardness of approximation for the minimum distance problem (following the framework given by~\cite{BartaIshaiOstrovskyWu2020}) but were ultimately unsuccessful. We then attempted to leverage a superconstant-factor NP hardness of approximation result for the shortest vector problem in higher $\ell_p$ norms~\cite{HairSahai2026SVP}, but again to no avail.\footnote{At the time, we did not use AI tools when pursuing these routes. We recently tested some of the public models on these ideas, but without success.} Inspired by a result using MinRank for public key encryption~\cite{ChatterjeeMuVasudevan2026}, we decided to explore MinRank in the context of witness encryption in the generic group model.

To aid in our search, we queried Codex / ChatGPT 6 Astra Ultra using harness components from the UCLA Moonshot Harness Project~\cite{MoonshotHarness}. Codex responded that it did not see how to use MinRank directly, but that the variant of MinRank presented here (bounded-factor MinRank) could be proven NP-hard and used to build witness encryption. Codex provided all of the major proof ideas along with a draft of the manuscript. We then digested the manuscript and made several revisions (sometimes with the help of Codex) to improve the presentation and give some additional intuition, especially in the technical overview (Section~\ref{sec:overview}) and Sections~\ref{sec:minrank}~and~\ref{sec:polylogminrank}.

Of course, the human authors take full responsibility for this paper and its contents. If any reader is aware of references that should be added, please contact us and we will update the manuscript.

We thank Brent Waters for useful discussions and pointers on applications of WE. This research was supported in part by a Laude Moonshot seed award, a Simons Investigator Award, a DARPA expMath award, NSF grant 2333935, BSF grant 2022370, a Xerox Faculty Research Award, a Google Faculty Research Award, an Okawa Foundation Research Grant, and the Symantec Chair of Computer Science.

\section{Preliminaries}\label{sec:prelim}
Vectors are columns. For a prime \(p\), we write \(\Fp=\mathbb Z/p\mathbb Z\);
all matrix ranks and spans are over \(\Fp\), unless stated otherwise.
A Boolean vector in \(\Fp^m\) is the image of a vector in \(\bits^m\);
its Hamming weight \(\wt(\mathbf v)\) counts its nonzero coordinates. We regard
these coordinates as the integers zero and one. A basis of a matrix space always means a linearly
independent basis.

\begin{definition}[Witness Encryption]\label{def:we}
Let $\mathcal R$ be a polynomial-time decidable, polynomially balanced
relation, and let $L=\{\varphi:\exists w,\;(\varphi,w)\in\mathcal R\}$.
A witness encryption scheme for $\mathcal R$ consists of polynomial-time
algorithms $\Enc$ and $\Dec$~\cite{GargGentrySahaiWaters2013,GentryLewkoWaters2014}.
Encryption takes $(1^\lambda,\varphi,\beta)$, where $\beta\in\bits$,
and outputs a ciphertext $C$. Decryption takes
$(1^\lambda,\varphi,w,C)$ and outputs a bit or $\bot$.
Writing $n=\max\{2,\lambda,|\varphi|\}$, we require:
\begin{enumerate}
\item \emph{Correctness.} There is a negligible function $\kappa$
such that, for every $\lambda$, every
$(\varphi,w)\in\mathcal R$, and every $\beta\in\bits$,
\[
 \Pr\!\left[
 \Dec(1^\lambda,\varphi,w,\Enc(1^\lambda,\varphi,\beta))=\beta
 \right]\ge 1-\kappa(n).
\]
\item \emph{False-statement security.} For every nonuniform probabilistic
polynomial-time adversary $A$, there is a negligible function
$\varepsilon_A$ such that, for every $\lambda$ and every
$\varphi\notin L$,
\[
 \begin{split}
 \Adv_A(\lambda,\varphi)
 =\bigl|&
 \Pr[A(1^\lambda,\varphi,\Enc(1^\lambda,\varphi,0))=1]\\
 &-\Pr[A(1^\lambda,\varphi,\Enc(1^\lambda,\varphi,1))=1]
 \bigr|
 \le\varepsilon_A(n).
 \end{split}
\]
\end{enumerate}
Probabilities are over the algorithms' randomness, and the bounds are
uniform over the statements and witnesses. No security requirement is
imposed when $\varphi\in L$.
\end{definition}

Our algorithms use the generic-group oracle below, with a fresh random encoding included in the probability space. Decryption and the adversary use the ciphertext's group context. We prove a quantitative bound on oracle queries allowing arbitrary local computation and advice subject to the restrictions below.

\paragraph{Statements and Witnesses.}
It suffices to treat Boolean circuit SAT: a polynomial-time NP verifier compiles into a circuit with polynomial overhead, and a satisfying input determines the gate values. We write $\varphi$ for this circuit and pad the algebraic encoding to at least $\lambda$ variables, with effective size $n=\max\{2,\lambda,|\varphi|\}$. Section~\ref{sec:together} checks the quantitative guarantees against the original statement size.

\paragraph{A Fresh Generic Cyclic Group.}

For each encryption and prime \(p\), the model supplies a fresh,
uniformly random injection
\[
 \sigma:\Fp\longrightarrow\bits^\ell,
 \qquad \ell=2\lceil\log p\rceil .
\]
We write $[a]=\sigma(a)$; in particular, $[1]$ is a generator and $[0]$ is the identity.  The group oracle takes
two strings and an addition or subtraction instruction, and returns
$[a+b]$ or $[a-b]$ on valid encodings $[a],[b]$;
it returns \(\bot\) if either operand is invalid.  This is an ordinary generic group: it's possible to multiply an encoded group element by a known scalar (using $O(\log p)$ addition/subtraction queries), but the oracle cannot compute $[ab]$ from $[a]$ and $[b]$.

A ciphertext includes access to its fresh group context $\Gamma$ and polynomially many handles (encoding strings); the encoding table belongs to the oracle model.

Unlike an abstract-handle model, ours exposes the encoding strings, a distinction that can affect security~\cite{Zhandry2022Labels}. The adversary may inspect them and use unbounded local computation. All oracle calls, including calls on guessed strings, count toward its query budget. Advice may depend arbitrarily on the statement and public prime, but not on the fresh encoding or hidden encryption coins. This excludes preprocessing of the encoding as studied by Coretti, Dodis, and Guo~\cite{CorettiDodisGuo2018}, while allowing arbitrary advice fixed before it is sampled.

\section{Witness Encryption From a Rank Gap}\label{sec:we}
We start by showing how a rank gap gives witness encryption. The reduction
needed here must preserve witnesses: a satisfying assignment must tell
us how to find the Boolean right factor, not merely imply its existence.

\begin{theorem}[Witness Encryption from Bounded-Factor MinRank]
\label{thm:we-from-minrank}
Let $n=\max\{2,\lambda,|\varphi|\}$. Suppose a deterministic
polynomial-time reduction computes an integer $m$ with
$16n\le m=\poly(n)$ and, given any prime $p\in[2^m,2^{m+1})$,
outputs an ordered basis $\mathbf M_1,\ldots,\mathbf M_k$ of a space
$\Smat\subseteq\Fp^{m\times m}$ with the following properties:
\begin{enumerate}
\item Every satisfying assignment to $\varphi$ efficiently yields
$\mathbf a\in\Fp^k$ and $\mathbf u,\mathbf v\in\Fp^m$ such that
\[
 \sum_{j=1}^k a_j\mathbf M_j=\mathbf u\mathbf v^T\ne0,
 \qquad \mathbf v\in\bits^m.
\]
\item If $\varphi$ is unsatisfiable, every nonzero matrix in $\Smat$
has rank at least $R=R(n)$.
\end{enumerate}
Assume $R\to\infty$ and $R\log n=o(m)$. Then, given any such prime
$p$ as an additional public input to encryption, there are polynomial-time
witness-encryption algorithms with correctness error at most $2^{-n}$.
For all sufficiently large $n$, every classical generic
adversary making at most $n^{R/8}$ oracle queries on a false statement
has distinguishing advantage at most $n^{-R/8}$.
\end{theorem}

In Section~\ref{sec:polylogminrank} we prove a version of
Theorem~\ref{thm:minrank-informal} showing that the required reduction exists for every polylogarithmic gap; in Section~\ref{sec:together} we
combine this reduction with Theorem~\ref{thm:we-from-minrank} to prove Theorem~\ref{thm:main}.
We prove the present theorem in the rest of this section.

\subsection{Parameters and Algorithms}
Given $(1^\lambda,\varphi)$, compute $m$, and let
$p\in[2^m,2^{m+1})$ be the supplied prime. Compute
the ordered basis $\mathbf M_1,\ldots,\mathbf M_k$; in particular,
$k\le m^2$. The basis depends only on the public input and $p$, so it
need not be included in the ciphertext.

We first describe the coordinate sampling steps exactly, and then bound
the time needed to implement them. If coordinate sampling fails, output
a fixed abort symbol.

\paragraph{Encryption $\Enc(1^\lambda,\varphi,\beta,p)$.}
Obtain a fresh group context $\Gamma$ of order $p$ and independently
sample
\[
 \mathbf s\leftarrow\Fp^m,\qquad
 \mathbf r\leftarrow\{1,\ldots,m\}^m,\qquad
 \mathbf t\leftarrow\Fp^k.
\]
Output
\begin{equation}\label{eq:ciphertext}
 \begin{split}
 C&=(p,\Gamma,[1],X_1,\ldots,X_m,Y_1,\ldots,Y_k),\\
 X_i&=[s_i],\qquad
 Y_j=\begin{cases}
  [\mathbf s^T\mathbf M_j\mathbf r],&\beta=0,\\
  [t_j],&\beta=1.
 \end{cases}
 \end{split}
\end{equation}
Both messages use the same samplers, including the unused randomness.

\paragraph{Decryption $\Dec(1^\lambda,\varphi,w,C)$.}
Check the witness. For a valid witness, compute the factors and
coefficients $\mathbf a,\mathbf u,\mathbf v$ promised by the reduction,
and compute
\[
 H=\sum_{i=1}^m u_i X_i,\qquad
 Z=\sum_{j=1}^k a_jY_j.
\]
Here addition and multiplication by a known scalar are generic-group
operations. Return zero if $Z$ is one of
$H,2H,\ldots,m^2H$, and return one otherwise. Compute this list by
successive additions. On the common abort outcome return zero; on an
invalid witness or malformed input return $\bot$.

All steps are polynomial time. Multiplication of a handle by a known
scalar costs $O(\log p)$ group queries, and the list has $m^2$ entries.
The ciphertext has $O(m+k)$ handles, each of $O(\log p)=O(m)$ bits,
so communication is polynomial as well.

\begin{lemma}[Bounded Sampling]\label{lem:sampling}
For every supplied prime $p$, the scheme has a bounded polynomial-time
implementation with sampling failure probability
$\delta\le2^{-n-2}$. Conditioned on successful sampling, the
distributions in~\eqref{eq:ciphertext} are exact.
The failure event is the same for both messages.
\end{lemma}
\begin{proof}
For each of the $2m+k$ coordinates, use rejection sampling with the minimum sufficient bit length and at most $n+2+\lceil\log(2m+k)\rceil$ attempts. Each attempt succeeds with probability at least $1/2$, so a union bound gives failure probability at most $2^{-n-2}$. Use independent randomness and the same schedule for both messages. Conditional on success, all coordinates are independent and exactly uniform.

\end{proof}

\subsection{Correctness}
\begin{proposition}[Correctness]\label{prop:correctness}
For every satisfying witness and either message, decryption fails with
probability at most
\[
 \delta+m^2 2^{-m}\le2^{-n}.
\]
\end{proposition}
\begin{proof}
Fix the supplied prime $p$ and condition on successful sampling.
For an encryption of zero,
\[
 Z=[\mathbf s^T(\textstyle\sum_j a_j\mathbf M_j)\mathbf r]
   =[(\mathbf s^T\mathbf u)(\mathbf v^T\mathbf r)]
   =(\mathbf v^T\mathbf r)H.
\]
Since $\mathbf u\mathbf v^T\ne0$, the Boolean vector $\mathbf v$ has
at least one nonzero coordinate. Each coordinate of $\mathbf r$ lies
in $\{1,\ldots,m\}$, so $1\le\mathbf v^T\mathbf r\le m^2$.
The tested list therefore contains $Z$, even
if $H=[0]$.

For an encryption of one, $\mathbf a\ne\mathbf0$, because its matrix
combination is nonzero. Hence $\mathbf a^T\mathbf t$ is uniform
in $\Fp$ and independent of $H$. The probability that
$Z=[\mathbf a^T\mathbf t]$ lies in the tested list is at most
$m^2/p$; repeated handles only shorten the list. Now use $p\ge2^m$
and $m\ge16n$ to bound this probability by
$2^{-n-1}$. Adding the sampling failure proves the claim.
\end{proof}
\subsection{Adaptive Generic Security}
Fix an unsatisfiable statement $\varphi$ and a supplied prime $p$, and
condition on successful sampling. We first consider a single linear
test, specified by a coefficient vector independently of the ciphertext.

\begin{lemma}[Fixed Linear Tests]\label{lem:collision}
For every $(\alpha,\boldsymbol{\beta},\boldsymbol{\gamma})\in
\Fp\times\Fp^m\times\Fp^k$, consider the test
\[
 \alpha[1]+\sum_i\beta_iX_i+\sum_j\gamma_jY_j=[0].
\]
For either encrypted bit, this test holds with probability one if its
coefficient vector is zero, and with probability at most
$m^{-R}+p^{-1}$ otherwise. For an encryption of one, the latter bound
improves to $p^{-1}$.
\end{lemma}
\begin{proof}
The zero coefficient vector always gives an identity. For any other
vector, an encryption of one gives the test
$\alpha+\boldsymbol{\beta}^T\mathbf s+
\boldsymbol{\gamma}^T\mathbf t=0$, a nonzero affine linear form in
uniform independent field elements. Its probability is at most $1/p$.

For an encryption of zero, put
$\mathbf M_{\boldsymbol{\gamma}}=\sum_j\gamma_j\mathbf M_j$.
The test becomes
\[
 \alpha+\mathbf s^T
 (\boldsymbol{\beta}+\mathbf M_{\boldsymbol{\gamma}}\mathbf r)=0.
\]
If $\boldsymbol{\gamma}=0$, the same affine-linear bound applies.
Otherwise, the basis is linearly independent, so
$\mathbf M_{\boldsymbol{\gamma}}\ne0$ and its rank $\rho$ is at least
$R$. Choose $\rho$ independent columns and fix the coordinates of
$\mathbf r$ outside these columns. At most one choice of the remaining
coordinates can satisfy
$\mathbf M_{\boldsymbol{\gamma}}\mathbf r=-\boldsymbol{\beta}$.
Thus this event has probability at most $m^{-\rho}\le m^{-R}$.
Outside this event, the test is a nonconstant affine linear equation in
uniform $\mathbf s$, and holds with probability $1/p$.
\end{proof}

In particular, for $R=(\log m)^C$, every fixed test either always holds
or has probability at most $m^{-(\log m)^C}+2^{-m}$, with the same
case applying to both messages. The passage from such fixed tests to
adaptive generic security is standard; see, for example,
Maurer~\cite[Appendix~B]{Maurer2005} and the random-encoding argument of Shoup~\cite{Shoup1997}. We give a short simulator proof to account for
the exposed encoding strings and arbitrary-string inputs in our model.

\begin{theorem}[Quantitative Generic Security]\label{thm:security}
Fix an unsatisfiable $\varphi$ and any supplied prime $p$.
Every classical generic adversary making $Q$ oracle queries has
distinguishing advantage at most
\begin{equation}\label{eq:security}
 O\!\left((Q+m+k)^2\bigl(m^{-R}+p^{-1}\bigr)\right).
\end{equation}
This allows arbitrary local computation and arbitrary-string oracle
inputs, and holds both conditional on successful sampling and with the
common abort outcome included.
\end{theorem}
\begin{proof}
The simulator represents each handle by the coefficient vector of an
affine form
$F=\alpha+\boldsymbol{\beta}^T\mathbf S+
\boldsymbol{\gamma}^T\mathbf Z$ in formal variables
$\mathbf S=(S_1,\ldots,S_m)^T$ and $\mathbf Z=(Z_1,\ldots,Z_k)^T$.
Initially the forms are $0,1,S_1,\ldots,S_m,Z_1,\ldots,Z_k$.
A group query adds or subtracts coefficient vectors. Equal vectors
receive the same label; each new vector receives a fresh uniform label
distinct from the previous labels. There are at most
$H=m+k+Q+2$ forms.

First suppose the adversary uses only assigned labels. Fix its coins,
its advice, and the simulator's labels. This fixes the entire adaptive
list of forms independently of the hidden encryption coins. Here the
advice may depend on $\varphi$ and $p$, but not on the fresh encoding or
those coins. The difference of two distinct forms is a nonzero
coefficient vector, so Lemma~\ref{lem:collision} bounds its collision
probability by $m^{-R}+p^{-1}$ for message zero and by $p^{-1}$ for
message one. Until such a collision occurs, either real experiment can
use the simulator's labels. A union bound over all pairs in both
experiments therefore gives advantage at most
$\binom H2(m^{-R}+2p^{-1})$. Fixing the symbolic list, rather than
conditioning a real transcript on previous noncollisions, accounts for
adaptivity even when choices depend on the labels' bit patterns.

For arbitrary-string inputs, the simulator rejects unassigned operands
and excludes them from future labels. At most $2Q$ strings are tested.
Until one is a valid unseen encoding, the remaining encodings are
uniform among unassigned, unexcluded strings. Since the label space
has size $U=2^\ell\ge p^2$, each new candidate is valid with
conditional probability at most $p/(U-H-2Q)$. We may assume
$Q+m+k\le p$, since otherwise~\eqref{eq:security} exceeds one.
The denominator is then $\Omega(p^2)$, so these guesses add only
$O(Q/p)$ to the advantage, which is absorbed in~\eqref{eq:security}.
We count either operand being a valid unseen encoding as a bad event,
so this also covers the oracle's undifferentiated invalid-input reply.
Finally, the common abort outcome has the same probability for both
messages and contributes no distinguishing advantage.
\end{proof}

\paragraph{From Rank Gap to Query Security.}
We finish the proof of Theorem~\ref{thm:we-from-minrank}. Since
$m=\poly(n)$ and $k\le m^2$, for $Q\le n^{R/8}$ the bound
in~\eqref{eq:security} is at most
\[
 n^{R/4+O(1)}\bigl(n^{-R}+2^{-m}\bigr).
\]
The first term is at most $\tfrac12 n^{-R/8}$ for sufficiently large
$n$ because $R\to\infty$, and the second satisfies the same bound
because $R\log n=o(m)$. Thus the advantage is at most $n^{-R/8}$.
Proposition~\ref{prop:correctness} gives correctness error at most
$2^{-n}$.

\section{High Level Introduction to the Rank Gap}\label{sec:overview}
In this overview, we focus on obtaining a rank gap of $1$ versus $3$. The same argument will give a logarithmic gap in Section~\ref{sec:minrank}; Section~\ref{sec:polylogminrank} then gives any prescribed polylogarithmic gap. We first recall the standard rank-one reduction and explain why the extension of Guruswami, Ren, and Tang~\cite{GuruswamiRenTang2026} gives only a constant gap under polynomial-time reductions. We then introduce our new approach.

\paragraph{A Gap of $1$ Versus $2$.}
Reduce a SAT instance to quadratic equations $q_1,\ldots,q_{N'}$ in $N$ variables over $\Fp$, including the Boolean equations $b_i^2-b_i=0$. For example, an AND gate gives $b_k-b_ib_j=0$, a NOT gate gives $b_k+b_i-1=0$, and acceptance gives $b_{\rm out}-1=0$. For each $\mathbf b\in\bits^N$, let
\[
 b_0=1,\qquad
 \mathbf v(\mathbf b)=(1,b_1,\ldots,b_N)^T,\qquad
 \mathbf M(\mathbf b)=\mathbf v(\mathbf b)\mathbf v(\mathbf b)^T.
\]
Consider the span of these matrices. Every matrix in this span is symmetric and satisfies $\mathbf M_{ii}=\mathbf M_{0i}$ for $i\ge1$. For $q(\mathbf b)=a_0+\sum_i a_ib_i+\sum_{i\le j}a_{ij}b_ib_j$, define the linear test
\begin{equation}\label{eq:test}
 \Test_q(\mathbf M)=a_0\mathbf M_{00}+\sum_i a_i\mathbf M_{0i}
                  +\sum_{i\le j}a_{ij}\mathbf M_{ij}.
\end{equation}
Then $\Test_q(\mathbf M(\mathbf b))=q(\mathbf b)$; for example, $b_4b_7-1=0$ becomes $\mathbf M_{47}-\mathbf M_{00}=0$. Impose these tests to obtain a linear subspace.

A satisfying assignment gives a nonzero rank-one matrix in this subspace. Conversely, every nonzero rank-one matrix in it has the form $\lambda\mathbf M(\mathbf b)$ for a satisfying assignment $\mathbf b$ and a nonzero scalar $\lambda$. To see this, suppose first that $\mathbf M_{00}=0$. The vanishing $2\times2$ minors on indices $0,i$ give $\mathbf M_{0i}=0$, and hence $\mathbf M_{ii}=0$. The minors on indices $i,j$ then give $\mathbf M_{ij}=0$, contradicting nonvanishing. We may therefore set $b_i=\mathbf M_{0i}/\mathbf M_{00}$. The minors on rows $0,i$ and columns $0,j$ give $\mathbf M_{ij}=\mathbf M_{00}b_ib_j$, and the diagonal identities give $b_i^2=b_i$. Finally, the linear tests imply that this Boolean assignment satisfies all the equations. Thus an unsatisfiable instance gives no nonzero rank-one matrix. This is the Boolean version of the quadratic-linearization reduction of Bl{\"a}ser, Ikenmeyer, Lysikov, Pandey, and Schreyer~\cite[Theorem~8.7]{BlaeserEtAl2019}.

\paragraph{Using More Equations.}
To obtain a larger gap, we would like to impose more consequences of the quadratic equations. Consider the following variant of the moment-matrix approach of Guruswami, Ren, and Tang~\cite{GuruswamiRenTang2026}. Choose a list $\mathcal H$ of polynomials, which we call \emph{weights}, and vertically stack the weighted tables:
\[
 \mathbf A(\mathbf b)=
 \begin{pmatrix}
 h_1(\mathbf b)\mathbf M(\mathbf b)\\
 h_2(\mathbf b)\mathbf M(\mathbf b)\\
 \vdots
 \end{pmatrix},\qquad h_1=1.
\]
Take the span of these stacked matrices, and impose every quadratic test on every block. In other words, a matrix $\mathbf A=\sum_{\mathbf b}\lambda_{\mathbf b}\mathbf A(\mathbf b)$ belongs to the final space $\Smat$ precisely when
\[
 \sum_{\mathbf b\in\bits^N}\lambda_{\mathbf b}h(\mathbf b)q_s(\mathbf b)=0
 \qquad\text{for every }h\in\mathcal H\text{ and }1\le s\le N'.
\]
The same coefficients $\lambda_{\mathbf b}$ occur in every block. A satisfying assignment still gives a rank-one matrix: all blocks have the same Boolean right factor $\mathbf v(\mathbf b)$. We can append zero columns to make the matrix square.

Using all monomials in the original bits up to a suitable degree gives a larger gap. With a degree cutoff of order $R$, this gives a gap of order $R$ and has size $N^{O(R)}$. For a logarithmic gap this is superpolynomial, so it is too expensive for our witness encryption scheme. Our goal is to find a much smaller list of weights that suffices for the same purpose.

\paragraph{The New Weights.}
For the $1$ versus $3$ gap, choose polynomially many pairs of affine linear forms $(\ell_{1,t},\ell_{2,t})$ in the bits. For each pair, use all monomials of degree at most two in these forms. Thus the part of an assignment matrix corresponding to $t$ is
\[
 \mathbf A^{(t)}(\mathbf b)=
 \begin{pmatrix}
 \mathbf M(\mathbf b)\\
 \ell_{1,t}(\mathbf b)\mathbf M(\mathbf b)\\
 \ell_{2,t}(\mathbf b)\mathbf M(\mathbf b)\\
 \ell_{1,t}(\mathbf b)\ell_{2,t}(\mathbf b)\mathbf M(\mathbf b)\\
 \ell_{1,t}(\mathbf b)^2\mathbf M(\mathbf b)\\
 \ell_{2,t}(\mathbf b)^2\mathbf M(\mathbf b)
 \end{pmatrix}.
\]
Stack these matrices over all $t$, take their span, and impose the weighted quadratic equations as before. Completeness is unchanged. For soundness, suppose the equations are unsatisfiable but the resulting space contains a nonzero matrix $\mathbf A$ of rank at most two.

We choose the form pairs so that, for some $t$, the corresponding submatrix $\mathbf A^{(t)}$ is nonzero and its two column combinations
\[
 \mathbf A^{(t)}\operatorname{vec}(\ell_{1,t}),\qquad
 \mathbf A^{(t)}\operatorname{vec}(\ell_{2,t})
\]
span its column space. Here $\operatorname{vec}(\ell)$ is the coefficient vector of $\ell$ with respect to $b_0,b_1,\ldots,b_N$. Section~\ref{sec:weights} constructs these pairs using the rank-preserving maps of Forbes and Shpilka~\cite{ForbesShpilka2012} and a polynomial root bound.

This submatrix must have rank exactly two. Otherwise, any nonzero block in it is a rank-one matrix passing all the quadratic tests, contradicting the $1$ versus $2$ argument. Fix such a $t$ and choose coefficients $\lambda_{\mathbf b}$ such that
\[
 \mathbf A^{(t)}
   =\sum_{\mathbf b\in\bits^N}\lambda_{\mathbf b}\mathbf A^{(t)}(\mathbf b).
\]
We keep these coefficients throughout the argument. To simplify notation, write $\mathbf A$ for this selected submatrix, $\ell_1,\ell_2$ for its two forms, and $\boldsymbol\ell(\mathbf b)=(\ell_1(\mathbf b),\ell_2(\mathbf b))$. We will show that every entry of $\mathbf A$ is zero, obtaining a contradiction.

\paragraph{Replacing the Original Bits.}
The two specified column combinations form a basis. Hence, for every $i\in\{0,\ldots,N\}$, there is a homogeneous linear polynomial $F_i$ in two variables such that
\[
 \mathbf A\operatorname{vec}\bigl(b_i-F_i(\boldsymbol\ell(\mathbf b))\bigr)=0.
\]
Each row of $\mathbf A$ gives an identity involving this difference. In the block with weight $H(\boldsymbol\ell(\mathbf b))$ and the row indexed by $j$, that identity is
\[
 \sum_{\mathbf b\in\bits^N}\lambda_{\mathbf b}
 \bigl(b_i-F_i(\boldsymbol\ell(\mathbf b))\bigr)b_jH(\boldsymbol\ell(\mathbf b))=0.
\]
Here $H$ is one of the six listed monomials. Taking linear combinations of rows and blocks gives the following \emph{replacement rule}:
\[
 \sum_{\mathbf b\in\bits^N}\lambda_{\mathbf b}
 \bigl(b_i-F_i(\boldsymbol\ell(\mathbf b))\bigr)a(\mathbf b)H(\boldsymbol\ell(\mathbf b))=0,
 \qquad \deg H\le2,
\]
where $a$ is any affine linear form in the original bits. Constants are allowed because $b_0=1$.

For example, consider row $2$, column $1$ of the block with weight $\ell_1\ell_2$. Its entry is the first sum below:
\[
 \begin{aligned}
 &\sum_{\mathbf b\in\bits^N}\lambda_{\mathbf b}
 b_2\ell_1(\mathbf b)\ell_2(\mathbf b)b_1\\
 &\quad=\sum_{\mathbf b\in\bits^N}\lambda_{\mathbf b}
 b_2\ell_1(\mathbf b)\ell_2(\mathbf b)F_1(\boldsymbol\ell(\mathbf b))\\
 &\quad=\sum_{\mathbf b\in\bits^N}\lambda_{\mathbf b}
 F_2(\boldsymbol\ell(\mathbf b))\ell_1(\mathbf b)\ell_2(\mathbf b)
 F_1(\boldsymbol\ell(\mathbf b)).
 \end{aligned}
\]
For the first equality, use the replacement rule with $i=1$, $a(\mathbf b)=b_2$, and $H(y_1,y_2)=y_1y_2$. For the second, use $i=2$ and $a(\mathbf b)=F_1(\boldsymbol\ell(\mathbf b))$, keeping the same $H$. This choice of $a$ is allowed because $F_1$ is linear and each $\ell_j$ is affine linear in the bits. In this way, every entry becomes a sum of the form $\sum_{\mathbf b}\lambda_{\mathbf b}f(\boldsymbol\ell(\mathbf b))$ with $\deg f\le4$. The same two replacements apply to each quadratic monomial of an input equation. Linear terms require one replacement, and constants remain unchanged.

\paragraph{Choosing Coordinates That Obey Boolean Identities.}
The values of $\ell_1$ and $\ell_2$ need not be Boolean. We therefore choose coordinates corresponding to two actual columns of $\mathbf A$ that form a basis, say the columns indexed by $i_1,i_2\in\{0,\ldots,N\}$. Set
\[
 w_j(\mathbf b)=F_{i_j}(\boldsymbol\ell(\mathbf b))\quad(j=1,2),
 \qquad \mathbf w(\mathbf b)=(w_1(\mathbf b),w_2(\mathbf b)).
\]
This is an invertible linear change of the two coordinates: both the chosen columns and the specified column combinations are bases. For formal variables $\mathbf y=(y_1,y_2)$, let $L_i(\mathbf y)$ be the linear polynomial satisfying
\[
 L_i(\mathbf w(\mathbf b))=F_i(\boldsymbol\ell(\mathbf b)).
\]
In particular, $L_{i_j}(\mathbf y)=y_j$. Since the coordinate change is invertible, every entry is now a sum $\sum_{\mathbf b}\lambda_{\mathbf b}f(\mathbf w(\mathbf b))$ with $\deg f\le4$, and every polynomial of degree at most two in $\mathbf w$ is an allowed weight combination.

We claim that the new coordinates satisfy Boolean identities inside these sums. Fix $j\in\{1,2\}$ and a polynomial $H$ of degree at most two. The replacement rule allows $b_{i_j}$ to be replaced by $w_j(\mathbf b)$ in both the linear and the quadratic term of $b_{i_j}^2-b_{i_j}$. Thus
\[
 \begin{aligned}
 &\sum_{\mathbf b\in\bits^N}\lambda_{\mathbf b}
 \bigl(w_j(\mathbf b)^2-w_j(\mathbf b)\bigr)H(\mathbf w(\mathbf b))\\
 &\quad=\sum_{\mathbf b\in\bits^N}\lambda_{\mathbf b}
 (b_{i_j}^2-b_{i_j})H(\mathbf w(\mathbf b))=0.
 \end{aligned}
\]
For the quadratic term, make the two replacements one at a time, just as in the example above. The last equality holds pointwise on the original assignments, also when $i_j=0$ since $b_0=1$.

Consequently, for a polynomial $f(\mathbf y)$ of degree at most four, we can replace every power $y_j^k$, $k\ge2$, by $y_j$ without changing the sum $\sum_{\mathbf b}\lambda_{\mathbf b}f(\mathbf w(\mathbf b))$. Indeed, each single reduction from $y_j^k$ to $y_j^{k-1}$ changes a monomial by $y_j^2-y_j$ times a polynomial of degree at most two, which the displayed identity permits. Call the resulting polynomial $f_{\rm ml}$, the \emph{multilinearization} of $f$: each variable has exponent at most one. It agrees with $f$ at every Boolean point. This argument concerns weighted sums; the individual values $w_j(\mathbf b)$ need not be bits.

\paragraph{Reducing to Four Boolean Points.}
For each point $\mathbf e\in\bits^2$, we use a polynomial that equals one at $\mathbf e$ and zero at the other three Boolean points. These four \emph{Boolean selectors} are
\begin{equation}\label{eq:selector-preview}
 E_{00}=(1-y_1)(1-y_2),\quad E_{01}=(1-y_1)y_2,
 \quad E_{10}=y_1(1-y_2),\quad E_{11}=y_1y_2.
\end{equation}
Every multilinear polynomial is determined by its values at the four points, so
\[
 f_{\rm ml}(\mathbf y)=\sum_{\mathbf e\in\bits^2}f(\mathbf e)E_{\mathbf e}(\mathbf y).
\]
Define the four field elements
\[
 \mu_{\mathbf e}=\sum_{\mathbf b\in\bits^N}\lambda_{\mathbf b}
 E_{\mathbf e}(\mathbf w(\mathbf b)),\qquad \mathbf e\in\bits^2.
\]
Multilinearizing and then using the selector expansion gives
\[
 \sum_{\mathbf b\in\bits^N}\lambda_{\mathbf b}f(\mathbf w(\mathbf b))
   =\sum_{\mathbf e\in\bits^2}\mu_{\mathbf e}f(\mathbf e),
 \qquad \deg f\le4.
\]
Since this formula applies to every matrix entry, it suffices to show that all four $\mu_{\mathbf e}$ vanish.

Fix $\mathbf e\in\bits^2$ and an input quadratic $q_s$. The selector has degree two, so $E_{\mathbf e}(\mathbf w(\mathbf b))$ is an allowed weight combination. The imposed equations and the replacement rule give
\[
 \begin{aligned}
 0&=\sum_{\mathbf b\in\bits^N}\lambda_{\mathbf b}
 E_{\mathbf e}(\mathbf w(\mathbf b))q_s(\mathbf b)\\
 &=\sum_{\mathbf b\in\bits^N}\lambda_{\mathbf b}
 E_{\mathbf e}(\mathbf w(\mathbf b))
 q_s(L_1(\mathbf w(\mathbf b)),\ldots,L_N(\mathbf w(\mathbf b))).
 \end{aligned}
\]
The polynomial in the last sum has degree at most four, so the selector expansion applies. All terms except the one at $\mathbf e$ vanish, giving
\[
 0=\mu_{\mathbf e}\,q_s(L_1(\mathbf e),\ldots,L_N(\mathbf e)).
\]
The values $(L_1(\mathbf e),\ldots,L_N(\mathbf e))$ form a candidate assignment. Since the original system is unsatisfiable, some quadratic fails at this assignment: either a Boolean equation fails, or all coordinates are bits and another input equation fails. For that quadratic, the displayed equality forces $\mu_{\mathbf e}=0$. This holds for every $\mathbf e$, so every entry of $\mathbf A$ is zero, the desired contradiction.

For a logarithmic gap, we use lists of $R$ forms and all their monomials of degree at most $R$, instead of pairs of forms and degree two. There are only $\binom{2R}{R}\le4^R$ such monomials per list, so $R=\Theta(\log N)$ still gives polynomial size. A matrix of rank at most $R$ admits at most $R$ replacement coordinates, and its Boolean selectors have degree at most $R$, allowing the same argument. The next section gives the details, including how to compute the matrix space without enumerating all $2^N$ assignments.

\section{Warm-Up: A Logarithmic Gap for Bounded-Factor MinRank}\label{sec:minrank}
In this section, we give a reduction showing superconstant-factor NP hardness of approximation for bounded-factor MinRank. We formalize the construction from Section~\ref{sec:overview}, using $R=\lfloor\log N\rfloor$ forms per list in place of two. This gives most of the ingredients needed in Section~\ref{sec:polylogminrank}. Combined with Theorem~\ref{thm:we-from-minrank} (and an algorithm to find the required prime), it already gives a witness encryption scheme where every generic adversary making at most $n^{c\log n}$ queries has distinguishing advantage at most $n^{-c\log n}$, for some constant $c > 0$.

\begin{theorem}[A Logarithmic Gap]\label{thm:warmup}
Given a SAT instance $\varphi$, one can compute an integer
$m=\poly(|\varphi|)$ such that, for every supplied prime $p\ge2^m$,
a deterministic polynomial-time algorithm outputs a basis of
$\Smat\subseteq\Fp^{m\times m}$ with the following properties:
\begin{enumerate}
\item Every satisfying assignment efficiently yields coefficients in
this basis and factors $\mathbf u,\mathbf v$ of a nonzero matrix
$\mathbf u\mathbf v^T\in\Smat$, with $\mathbf v\in\bits^m$.
\item If $\varphi$ is unsatisfiable, every nonzero matrix in $\Smat$
has rank at least a fixed positive constant times $\log m$.
\end{enumerate}
The running time is polynomial in $|\varphi|$ and $\log p$.
\end{theorem}

\subsection{The Assignment Tables and Weighted Constraints}\label{sec:tables}\label{sec:stack}
We use the quadratic system and assignment tables from Section~\ref{sec:overview}. Circuit SAT gives $N'$ quadratic equations $q_s(\mathbf b)=0$ in $N$ Boolean variables with polynomial overhead, and a satisfying input determines all gate bits efficiently. Pad with unused bits so that $N\ge\max\{2,|\varphi|\}$, and set $R=\lfloor\log N\rfloor$. Recall that $b_0=1$ and
\[
 \mathbf v(\mathbf b)=(1,b_1,\ldots,b_N)^T,
 \qquad \mathbf M(\mathbf b)=\mathbf v(\mathbf b)\mathbf v(\mathbf b)^T.
\]
Tables of monomial products are also used in relaxations for Boolean optimization. The methods of Sherali and Adams, Lov\'asz and Schrijver, and Lasserre introduce variables for such products and impose linear inequalities or positive semidefinite constraints~\cite{SheraliAdams1990Hierarchy,LovaszSchrijver1991Cones,Lasserre2001Moments}; see Laurent~\cite{Laurent2003} for a comparison. Here we use equality constraints over $\Fp$ and control matrix rank.

For a finite list of weights $\mathcal H=(h_1,h_2,\ldots)$ containing $1$, define
\begin{equation}\label{eq:assignment-stack}
 \mathbf A(\mathbf b)=\begin{pmatrix}
 h_1(\mathbf b)\mathbf M(\mathbf b)\\
 h_2(\mathbf b)\mathbf M(\mathbf b)\\
 \vdots
 \end{pmatrix}
 =\begin{pmatrix}h_1(\mathbf b)\mathbf v(\mathbf b)\\h_2(\mathbf b)\mathbf v(\mathbf b)\\\vdots\end{pmatrix}\mathbf v(\mathbf b)^T,
 \qquad
 \Vmat=\Span\{\mathbf A(\mathbf b):\mathbf b\in\bits^N\}.
\end{equation}
For $\mathbf A=\sum_{\mathbf b}\lambda_{\mathbf b}\mathbf A(\mathbf b)$, write its block with weight $h$ as
\begin{equation}\label{eq:block}
 \mathbf B_h=\sum_{\mathbf b}\lambda_{\mathbf b}h(\mathbf b)\mathbf M(\mathbf b).
\end{equation}
Using the tests in~\eqref{eq:test}, our output space is
\begin{equation}\label{eq:space}
 \Smat=\{\mathbf A\in\Vmat:\Test_{q_s}(\mathbf B_h)=0
       \text{ for all }1\le s\le N',\ h\in\mathcal H\}.
\end{equation}
Equivalently, these homogeneous linear constraints require
\begin{equation}\label{eq:weighted-equations}
 \sum_{\mathbf b}\lambda_{\mathbf b}h(\mathbf b)q_s(\mathbf b)=0.
\end{equation}
Related constraints appear in the work of Khot and Saket~\cite{KhotSaket2017}, who study satisfaction of quadratic equations over $\mathbb F_2$ in \emph{superposition}: for an equation with no constant term, the values on a list of assignments must sum to zero. Guruswami, Ren, and Tang~\cite{GuruswamiRenTang2026} use this viewpoint for MinRank hardness. Our constraints allow arbitrary coefficients in $\Fp$ and also multiply each equation by polynomial weights.

A satisfying assignment gives a nonzero rank-one matrix $\mathbf A(\mathbf b)\in\Smat$ with Boolean right factor $\mathbf v(\mathbf b)$. It remains to choose a polynomial-size list of weights that rules out nonzero matrices of rank at most $R$ on an unsatisfiable instance.

\subsection{Choosing Weights That Work for Every Possible Column Space}\label{sec:weights}

For the threshold $R=\lfloor\log N\rfloor$ fixed above, build weights from $R$ linear
combinations of the coordinates $b_0,b_1,\ldots,b_N$, remembering that
$b_0=1$. If these combinations are $\ell_1(\mathbf b),\ldots,\ell_R(\mathbf b)$,
include all products
\begin{equation}\label{eq:generic-weights}
 \ell_1(\mathbf b)^{\alpha_1}\cdots\ell_R(\mathbf b)^{\alpha_R},
 \qquad \alpha_j\in\mathbb Z_{\ge0},\quad \alpha_1+\cdots+\alpha_R\le R.
\end{equation}
Thus we can take any polynomial of degree at most $R$ in these $R$
linear forms as a weight, by linearly combining blocks.
There are only $\binom{2R}{R}\le4^R$ listed monomials.

We use $R$ combinations because a matrix of rank at most $R$ has a
column space of dimension at most $R$. We want the corresponding column combinations to
span that space. The choice must be made before seeing the matrix, so
we include several lists of combinations. The next elementary lemma
will ensure that one list works.

\paragraph{A Family of Spanning Combinations.}

Require the supplied prime $p$ to satisfy the bound below, and set
\begin{equation}\label{eq:prime}
 p>\max\{2^N,\,2NR\},
 \qquad D=NR,\qquad S=\{0,1,\ldots,2D\}\subseteq\Fp.
\end{equation}
The exponential lower bound on $p$ is harmless for the reduction: the
running time is polynomial in $\log p$. Its purpose is to make the
distinct integer powers $1,2,4,\ldots,2^N$ remain distinct in the field.

For each $t\in S$ use the $R$ forms
\begin{equation}\label{eq:forms}
 \ell_{j,t}(\mathbf b)=\sum_{i=0}^N (2^{j-1}t)^i b_i,
 \qquad 1\le j\le R.
\end{equation}
As usual, a zeroth power is one, including when $t=0$.
For example, the first two combinations, when $R\ge2$, are
\[
 \ell_{1,t}(\mathbf b)=1+\sum_{i=1}^N t^i b_i,
 \qquad
 \ell_{2,t}(\mathbf b)=1+\sum_{i=1}^N (2t)^i b_i.
\]
The same coefficients will be applied to matrix columns in the next lemma.

We will use the elementary \emph{polynomial root bound}: a nonzero polynomial of degree $d$ over a field has at most $d$ distinct roots. The spanning family in the lemma below is a specialization of the rank-preserving maps of Forbes and Shpilka~\cite[Theorem~4.1]{ForbesShpilka2012}, which refine an earlier construction of Gabizon and Raz~\cite{GabizonRaz2008}. These maps preserve the dimension of a given low-dimensional space for all but a small set of parameter values. A family of maps into a smaller ambient space with this dimension-preservation property is called a \emph{rank condenser}. The same evaluation construction underlies later work on rank condensers and on controlling how a fixed subspace intersects a collection of other subspaces~\cite{GuruswamiKopparty2016,ForbesSaptharishiShpilka2014,ForbesGuruswami2015}.

\begin{lemma}\label{lem:spanning}
Assume $p>2^N$. Let $\mathbf c_0,\ldots,\mathbf c_N$ be vectors over $\Fp$ spanning a space of dimension
$1\le\rho\le N+1$. Define vector polynomials
\[
 \mathbf d_j(T)=\sum_{i=0}^N(2^{j-1}T)^i \mathbf c_i,
 \qquad 1\le j\le\rho.
\]
There is a nonzero scalar polynomial $\Delta(T)$ of degree at most
$N\rho$ such that, whenever $\Delta(t)\ne0$, the vectors
$\mathbf d_1(t),\ldots,\mathbf d_\rho(t)$ are a basis of the given space.
\end{lemma}
\begin{proof}
Choose coordinates on the $\rho$-dimensional space. The coordinates of
$\sum_i \mathbf c_iT^i$ are $\rho$ polynomials $f_1(T),\ldots,f_\rho(T)$,
each of degree at most $N$. They are linearly independent. To see why,
suppose scalars $r_1,\ldots,r_\rho$ give
$\sum_j r_j f_j(T)=0$. Comparing the coefficient of each $T^i$ says
that the row vector $(r_1,\ldots,r_\rho)$ times the coordinate vector
of $\mathbf c_i$ is zero. The $\mathbf c_i$ span the whole coordinate space, so that row
vector must be zero.

Gaussian elimination on their coefficient vectors lets us replace these
polynomials by invertible linear combinations with distinct degrees
$e_1<\cdots<e_\rho$. One way to do this is to select a polynomial
of highest degree, cancel its leading coefficient in the others, and
continue. Denote the resulting nonzero leading coefficients by
$a_1,\ldots,a_\rho$.

Consider the matrix with entry $f_i(2^{j-1}T)$ in row $i$, column $j$,
after this elimination. Each term in its determinant takes one entry
from every row, so its degree is at most $\sum_i e_i$.
To obtain degree exactly $\sum_i e_i$, we must take the leading term
from every chosen entry. The leading term in row $i$, column $j$ is
$a_i2^{(j-1)e_i}T^{e_i}$. Thus the coefficient of $T^{\sum_i e_i}$ is
\[
 \left(\prod_{i=1}^\rho a_i\right)
 \det\begin{pmatrix}
 1&2^{e_1}&(2^{e_1})^2&\cdots&(2^{e_1})^{\rho-1}\\
 1&2^{e_2}&(2^{e_2})^2&\cdots&(2^{e_2})^{\rho-1}\\
 \vdots&\vdots&\vdots&&\vdots\\
 1&2^{e_\rho}&(2^{e_\rho})^2&\cdots&(2^{e_\rho})^{\rho-1}
 \end{pmatrix}.
\]
The displayed matrix is invertible. Indeed, a nontrivial linear combination
of its columns equal to zero would describe a nonzero polynomial of
degree at most $\rho-1$ vanishing at the $\rho$ distinct points
$2^{e_1},\ldots,2^{e_\rho}$. This is impossible. These points are
distinct because $e_i\le N$ and $p>2^N$.

Thus the determinant is a nonzero polynomial of degree at most $N\rho$.
Undoing the invertible row operations changes it only by a nonzero
constant. Take the original determinant as $\Delta(T)$.
Its nonvanishing says precisely that the $\mathbf d_j(t)$ are independent.
\end{proof}

\paragraph{The Complete Weight List and the Resulting Size.}

For every $t\in S$ and every exponent tuple
$\boldsymbol{\alpha}=(\alpha_1,\ldots,\alpha_R)$ in~\eqref{eq:generic-weights},
put the weight
\begin{equation}\label{eq:final-weight}
 h_{\boldsymbol{\alpha},t}(\mathbf b)=\prod_{j=1}^R\ell_{j,t}(\mathbf b)^{\alpha_j}
\end{equation}
in $\mathcal H$. Keep repeated weights as separate blocks; this makes
the description uniform and does not harm the proof.
Construct $\mathbf A(\mathbf b)$, $\Vmat$, and $\Smat$ as in
\eqref{eq:assignment-stack}--\eqref{eq:space}.
There are
\[
 K=(2NR+1)\binom{2R}{R}
\]
blocks, so the matrices have $m=(N+1)K$ rows and $N+1$ columns.
Append $m-(N+1)$ zero columns to make them square. This does not change
their rank. The right factor of an assignment matrix is then
$(1,b_1,\ldots,b_N,0,\ldots,0)^T$, a Boolean vector with at most
$N+1$ nonzero entries.

\begin{proposition}\label{prop:construction}
For the supplied prime in~\eqref{eq:prime}, the construction gives a
matrix space $\Smat\subseteq\Fp^{m\times m}$ with
\begin{equation}\label{eq:m}
 m=(N+1)(2NR+1)\binom{2R}{R}.
\end{equation}
A satisfying assignment gives a nonzero rank-one matrix in $\Smat$
with the Boolean right factor just described. If the equations are
unsatisfiable, every nonzero matrix in $\Smat$ has rank greater than $R$.
An ordered basis of $\Smat$ is deterministically computable in time
polynomial in $N,N',m,\log p$. Given a satisfying assignment, the
coordinates of its matrix in that basis are computable within the same bound.
For our choice $R=\lfloor\log N\rfloor$, one has
$m=O(N^4\log N)$, and every nonzero NO-case matrix has rank
$\Omega(\log m)$.
\end{proposition}

We already proved the satisfying-assignment claim immediately after
\eqref{eq:space}. The next subsection proves the rank lower bound;
Section~\ref{sec:algorithm} gives an efficient algorithm for constructing the basis.

\subsection{Soundness: A Nonzero Low-Rank Matrix Yields a Solution}\label{sec:soundness}

Suppose, toward a contradiction, that the input equations are
unsatisfiable but $\mathbf A\in\Smat$ is nonzero and has rank
$1\le\rho\le R$. Discard its appended zero columns. Choose an
expression
\[
 \mathbf A=\sum_{\mathbf b\in\bits^N}\lambda_{\mathbf b} \mathbf A(\mathbf b)
\]
and keep these same coefficients throughout the proof. This expression
need not be unique and we do not need to find it algorithmically.
Assume $p>2^N$. Let $\mathbf c_0,\ldots,\mathbf c_N$ be the actual column vectors of this whole stack.

\paragraph{Step 1: Choosing a Useful List of Forms.}

Apply Lemma~\ref{lem:spanning} to the columns $\mathbf c_i$, obtaining
$\Delta(T)$. Its nonvanishing will give a basis of the whole column
space. We also want a nonzero block among the weights at the same value
of $t$.

Since $\mathbf A\ne0$, some block entry is nonzero. Suppose it has weight
index $\boldsymbol{\alpha}$, sample value $t_0$, and row and column indices $i,j$.
Replace the sample value in that entry by a formal variable $T$:
\[
 f(T)=\sum_{\mathbf b}\lambda_{\mathbf b} b_i b_j
       \prod_{s=1}^R
       \left(\sum_{k=0}^N(2^{s-1}T)^k b_k\right)^{\alpha_s}.
\]
This polynomial is nonzero because $f(t_0)\ne0$.
Each factor inside parentheses has degree at most $N$, and
$\sum_s\alpha_s\le R$, so $\deg f\le NR=D$.
The product $f(T)\Delta(T)$ is therefore a nonzero polynomial of degree
at most $2D$. It cannot vanish at every one of the $2D+1$ distinct
values in $S$.

Fix a value $t$ where the product is nonzero. At this value,
\begin{enumerate}
\item the first $\rho$ column combinations in Lemma~\ref{lem:spanning}
form a basis of the entire column space of $\mathbf A$;
\item at least one of the blocks with this value of $t$ is nonzero.
\end{enumerate}
We will show that all of the blocks with this value of $t$ are zero, obtaining a contradiction.
From now on this $t$ is fixed. Write $\ell_j=\ell_{j,t}$ and
$\boldsymbol{\ell}(\mathbf b)=(\ell_1(\mathbf b),\ldots,\ell_R(\mathbf b))$.

\paragraph{Step 2: Replacements from Column Relations.}

Suppose scalars $\kappa_0,\ldots,\kappa_N\in\Fp$ give a column
relation $\kappa_0\mathbf c_0+\cdots+\kappa_N\mathbf c_N=0$, and write
$k(\mathbf b)=\sum_{i=0}^N\kappa_i b_i$.
Let $a(\mathbf b)$ be any linear combination of $b_0,\ldots,b_N$, and let
$H$ be a polynomial of degree at most $R$ in $R$ variables.
Expanding block rows gives
\begin{equation}\label{eq:kernel-one}
 \sum_{\mathbf b}\lambda_{\mathbf b} k(\mathbf b)a(\mathbf b)H(\boldsymbol{\ell}(\mathbf b))=0,
 \quad \deg H\le R.
\end{equation}
Indeed, the individual rows give $a=b_i$; linear combinations of rows
give general $a$. Linear combinations of the monomial blocks give
general $H$. Constants are allowed as $a$ because $b_0=1$.

For any polynomial $G$ of degree at most $R+1$ in $R$ variables,
we also have
\begin{equation}\label{eq:kernel-pure}
 \sum_{\mathbf b}\lambda_{\mathbf b} k(\mathbf b)G(\boldsymbol{\ell}(\mathbf b))=0,
 \quad \deg G\le R+1.
\end{equation}
For a nonconstant monomial in $G$, use one of its $\ell_j$ factors as
$a$ in~\eqref{eq:kernel-one}; the remaining product has degree at most
$R$. For a constant use $a=1$. Adding these identities proves the claim.
These two displayed equations are the replacement rules we need.

\paragraph{Step 3: Choosing $\rho$ Replacement Coordinates.}

Choose $\rho$ independent columns $\mathbf c_{i_1},\ldots,\mathbf c_{i_\rho}$ from
$\mathbf c_0,\ldots,\mathbf c_N$. They form a basis of the column space. Introduce
formal variables $y_1,\ldots,y_\rho$, one for each chosen column.
For each $i$, record the coefficients expressing $\mathbf c_i$ in this basis
in a linear polynomial:
\[
 \mathbf c_i=\sum_{j=1}^\rho \gamma_{ij}\mathbf c_{i_j},
 \qquad L_i(\mathbf y)=\sum_{j=1}^\rho \gamma_{ij}y_j,
 \qquad \mathbf y=(y_1,\ldots,y_\rho).
\]
For a chosen basis column its own coordinate is one and the others
are zero, so $L_{i_j}(\mathbf y)=y_j$.

For example, suppose $\rho=2$, columns $\mathbf c_0,\mathbf c_2$ form a basis, and
$\mathbf c_1=3\mathbf c_0+2\mathbf c_2$. Then
\[
 L_0(\mathbf y)=y_1,\qquad L_2(\mathbf y)=y_2,\qquad L_1(\mathbf y)=3y_1+2y_2.
\]
In general, column zero may or may not be among the chosen columns.

Next we need formulas for these chosen coordinates using the available
weight forms. This will ensure that polynomials of degree at most $R$
in the new coordinates are allowed as weights.
The first $\rho$ combinations
$\mathbf d_j=\sum_i(2^{j-1}t)^i \mathbf c_i$ are another basis, by our choice of $t$.
Express each chosen actual column in that basis:
\[
 \mathbf c_{i_j}=\sum_{s=1}^\rho u_{js}\mathbf d_s,
 \qquad w_j(\mathbf b)=\sum_{s=1}^\rho u_{js}\ell_s(\mathbf b),
 \qquad \mathbf w(\mathbf b)=(w_1(\mathbf b),\ldots,w_\rho(\mathbf b)).
\]
The variables $y_j$ are formal; throughout the proof, we evaluate polynomials in
$\mathbf y$ at the field values $\mathbf y=\mathbf w(\mathbf b)$.
Replacing every $b_k$ in the formula for $w_j(\mathbf b)$ by the
column $\mathbf c_k$ gives
$\sum_{s=1}^{\rho}u_{js}\mathbf d_s=\mathbf c_{i_j}$.
Making the same substitution in $L_i(\mathbf w(\mathbf b))$ therefore gives
\[
 \sum_{j=1}^{\rho}\gamma_{ij}
       \left(\sum_{s=1}^{\rho}u_{js}\mathbf d_s\right)
 =\sum_{j=1}^{\rho}\gamma_{ij}\mathbf c_{i_j}
 =\mathbf c_i.
\]
Thus the coefficients of $b_i-L_i(\mathbf w(\mathbf b))$, viewed as a
linear combination of $b_0,\ldots,b_N$, give a combination of columns
equal to zero. Both replacement rules in Step 2 therefore apply to
each of these differences.

\emph{Note that the values $w_j(\mathbf b)$ need not be Boolean.} They are built from
linear combinations of the weight forms, not by taking bits of $\mathbf b$.
Likewise, $b_i=L_i(\mathbf w(\mathbf b))$ need not hold pointwise. The equations in
Step 2 say precisely where these replacements are valid.

\paragraph{Step 4: Replacing the Bits in a Weighted Quadratic.}

Let $g(x_0,\ldots,x_N)$ be any polynomial of degree at most two and
let $H$ be any polynomial of degree at most $R$ in $R$ variables. We claim
\begin{equation}\label{eq:replacement}
 \sum_{\mathbf b}\lambda_{\mathbf b} g(\mathbf v(\mathbf b))H(\boldsymbol{\ell}(\mathbf b))
 =\sum_{\mathbf b}\lambda_{\mathbf b}
 g\bigl(L_0(\mathbf w(\mathbf b)),\ldots,L_N(\mathbf w(\mathbf b))\bigr)H(\boldsymbol{\ell}(\mathbf b)).
\end{equation}
The weight $H(\boldsymbol{\ell}(\mathbf b))$ is the same on both sides. Only the inputs to
$g$ have been replaced. Constants in $g$ remain constants.
The occurrence of $\mathbf v(\mathbf b)$ simply
means that the zeroth input of $g$ is evaluated at $b_0=1$.

For a quadratic monomial, first replace $b_i$ by
$L_i(\mathbf w(\mathbf b))$, and then replace $b_j$ by
$L_j(\mathbf w(\mathbf b))$. The two changes are expressed by the identity
\[
 b_i b_j-L_i(\mathbf w(\mathbf b))L_j(\mathbf w(\mathbf b))
 =(b_i-L_i(\mathbf w(\mathbf b)))b_j
   +L_i(\mathbf w(\mathbf b))(b_j-L_j(\mathbf w(\mathbf b))).
\]
Multiply by $H(\boldsymbol{\ell}(\mathbf b))$ and sum with coefficients $\lambda_{\mathbf b}$.
Both weighted sums vanish by~\eqref{eq:kernel-one}: the other factor is a
linear combination of the original $b_k$, and $\deg H\le R$.
A linear monomial is treated with the same rule and $a=1$;
a constant needs no change. Adding these monomial identities proves
\eqref{eq:replacement}.

We will use weights that are polynomials in the new coordinates.
Given $F(\mathbf y)$ of degree at most $R$, choose $H$ so that
$H(\boldsymbol{\ell}(\mathbf b))=F(\mathbf w(\mathbf b))$, by substituting the formulas defining $w_j$.
Since each $w_j$ is a linear combination of the $\ell_s$, this $H$
still has degree at most $R$. Equation~\eqref{eq:replacement} therefore
allows us to replace $g$ while keeping the multiplier $F(\mathbf w(\mathbf b))$ unchanged.

For an input equation, use the polynomial
$g(x_0,\ldots,x_N)=q_s(x_1,\ldots,x_N)$, which simply ignores
its zeroth argument. Since $F(\mathbf w(\mathbf b))$ is a linear
combination of the listed weights at $t$, the imposed tests
in~\eqref{eq:weighted-equations} give
\[
 \sum_{\mathbf b}\lambda_{\mathbf b}
 q_s(\mathbf b)F(\mathbf w(\mathbf b))=0.
\]
Equation~\eqref{eq:replacement} replaces the arguments of $q_s$
without changing this sum. For a Boolean identity, use $g=x_i^2-x_i$;
the original value is zero also when $i=0$, because $b_0=1$.
Applying this replacement to these two kinds of polynomial gives,
for every $F$ with $\deg F\le R$,
\begin{align}
 \sum_{\mathbf b}\lambda_{\mathbf b}
 q_s\bigl(L_1(\mathbf w(\mathbf b)),\ldots,L_N(\mathbf w(\mathbf b))\bigr)F(\mathbf w(\mathbf b))&=0,
                                                   \label{eq:source-reduced}\\
 \sum_{\mathbf b}\lambda_{\mathbf b}
 \bigl(L_i(\mathbf w(\mathbf b))^2-L_i(\mathbf w(\mathbf b))\bigr)F(\mathbf w(\mathbf b))&=0,
                 \qquad 0\le i\le N.              \label{eq:boolean-reduced}
\end{align}
In particular, since $L_{i_j}(\mathbf y)=y_j$, the second line gives
\begin{equation}\label{eq:coordinate-boolean}
 \sum_{\mathbf b}\lambda_{\mathbf b}\bigl(w_j(\mathbf b)^2-w_j(\mathbf b)\bigr)F(\mathbf w(\mathbf b))=0,
 \qquad \deg F\le R.
\end{equation}
Thus the new coordinates obey Boolean
identities inside the weighted sums. We have not assumed that their
individual values are bits.

\paragraph{Step 5: Applying the Selectors.}

We can now use the selectors from Section~\ref{sec:overview}, this time on the formal
variables $y_1,\ldots,y_\rho$. Explicitly, for $\mathbf e\in\bits^\rho$, set
\[
 E_{\mathbf e}(\mathbf y)=\prod_{j:e_j=1}y_j\prod_{j:e_j=0}(1-y_j).
\]
To apply these selectors inside the weighted sums, we first reduce the polynomials to multilinear form; the values $\mathbf w(\mathbf b)$ need not be Boolean.

For any polynomial $f(\mathbf y)$ of degree at most $R+2$, repeatedly replace
every power $y_j^k$, $k\ge2$, by $y_j^{k-1}$. Each change is a multiple
of $y_j^2-y_j$, because
\[
 y_j^k-y_j^{k-1}=(y_j^2-y_j)y_j^{k-2}.
\]
Including any other factors of the monomial, the multiplier has degree
at most $(R+2)-2=R$. Equation~\eqref{eq:coordinate-boolean} therefore
shows that these changes preserve the sum
$\sum_{\mathbf b}\lambda_{\mathbf b} f(\mathbf w(\mathbf b))$.
Call the result $f_{\rm ml}$. It is a \emph{multilinear} polynomial:
every variable has exponent at most one. It has the same values as
$f$ on $\bits^\rho$.

Every multilinear polynomial equals the sum of its Boolean values times
the corresponding selectors. In one variable this is
$a+cy=a(1-y)+(a+c)y$; applying this identity to one variable after
another gives
\[
 f_{\rm ml}(\mathbf y)=\sum_{\mathbf e\in\bits^\rho} f(\mathbf e)E_{\mathbf e}(\mathbf y).
\]
Consequently, if we define the selector weights
\[
 \mu_{\mathbf e}=\sum_{\mathbf b}\lambda_{\mathbf b} E_{\mathbf e}(\mathbf w(\mathbf b)),\qquad \mathbf e\in\bits^\rho,
\]
then for every $f$ of degree at most $R+2$ we have
\begin{equation}\label{eq:cube}
 \sum_{\mathbf b}\lambda_{\mathbf b} f(\mathbf w(\mathbf b))
       =\sum_{\mathbf e\in\bits^\rho}\mu_{\mathbf e} f(\mathbf e).
\end{equation}
Thus the sums we need can be computed using only $2^\rho$ field weights
on a Boolean cube. This conclusion is about those sums; it is not a
claim that the original coefficients $\lambda_{\mathbf b}$ have small support.

Fix a cube point $\mathbf e\in\bits^\rho$. Consider the candidate assignment
\[
 (L_1(\mathbf e),\ldots,L_N(\mathbf e)).
\]
This cannot be a satisfying Boolean assignment. Either some $L_i(\mathbf e)$ is not a bit, so
$L_i(\mathbf e)^2-L_i(\mathbf e)\ne0$, or all are bits and one of the input equations
fails. Therefore some polynomial $Q(\mathbf y)$ from the list
\[
 \{L_i(\mathbf y)^2-L_i(\mathbf y):1\le i\le N\}
 \ \cup\
 \{q_s(L_1(\mathbf y),\ldots,L_N(\mathbf y)):1\le s\le N'\}
\]
has $Q(\mathbf e)\ne0$.
Equations~\eqref{eq:source-reduced}--\eqref{eq:boolean-reduced} apply
with multiplier $E_{\mathbf e}$, since $\deg E_{\mathbf e}=\rho\le R$. They give
\begin{equation}\label{eq:chain1}
 0=\sum_{\mathbf b}\lambda_{\mathbf b} Q(\mathbf w(\mathbf b))E_{\mathbf e}(\mathbf w(\mathbf b)).
\end{equation}
Since $\deg(QE_{\mathbf e})\le\rho+2\le R+2$, the multilinearization
described above gives a multilinear polynomial $U(\mathbf y)$ that agrees with
$Q(\mathbf y)E_{\mathbf e}(\mathbf y)$ on all of $\bits^\rho$ and satisfies
\begin{equation}\label{eq:chain2}
 \sum_{\mathbf b}\lambda_{\mathbf b} Q(\mathbf w(\mathbf b))E_{\mathbf e}(\mathbf w(\mathbf b))=\sum_{\mathbf b}\lambda_{\mathbf b}U(\mathbf w(\mathbf b)).
\end{equation}
Applying~\eqref{eq:cube} to $U$, and using that $E_{\mathbf e}$ vanishes
at every cube point other than $\mathbf e$, we get
\begin{equation}\label{eq:chain3}
    \sum_{\mathbf b}\lambda_{\mathbf b} U(\mathbf w(\mathbf b))
   =\mu_{\mathbf e} Q(\mathbf e).
\end{equation}
Combining~\eqref{eq:chain1}--\eqref{eq:chain3}, we have $0 = \mu_{\mathbf e}Q(\mathbf e)$.
Since $Q(\mathbf e)\ne0$, it follows that $\mu_{\mathbf e}=0$.
This holds for every $\mathbf e$. Equation~\eqref{eq:cube} now says
\begin{equation}\label{eq:all-zero}
 \sum_{\mathbf b}\lambda_{\mathbf b} f(\mathbf w(\mathbf b))=0\qquad\text{for every }\deg f\le R+2.
\end{equation}

\paragraph{Step 6: Returning to the Matrix Entries.}

An entry of any block at the fixed sample is
$\sum_{\mathbf b}\lambda_{\mathbf b} b_i b_j H(\boldsymbol{\ell}(\mathbf b))$ for a weight monomial $H$ of
degree at most $R$. Equation~\eqref{eq:replacement} changes this to
\[
 \sum_{\mathbf b}\lambda_{\mathbf b} L_i(\mathbf w(\mathbf b))L_j(\mathbf w(\mathbf b))H(\boldsymbol{\ell}(\mathbf b)).
\]
To apply~\eqref{eq:all-zero}, we must express the remaining factors
$\ell_s(\mathbf b)$ in terms of $\mathbf w(\mathbf b)$ as well. For each such factor, make the
replacement
\[
 \ell_s(\mathbf b)=\sum_{k=0}^N(2^{s-1}t)^k b_k
 \quad\longmapsto\quad
 \sum_{k=0}^N(2^{s-1}t)^k L_k(\mathbf w(\mathbf b)).
\]
The difference is a linear combination of $b_k-L_k(\mathbf w(\mathbf b))$, so its
coefficients give a column relation. All the other factors, including
ones already replaced, are polynomials in the original $\ell_s$:
each $w_j$ is a linear combination of those forms. The product of
the other factors has degree at most $2+(R-1)=R+1$ in these forms.
Equation~\eqref{eq:kernel-pure} therefore
justifies each replacement inside the sum.

After these replacements the entry has the form
$\sum_{\mathbf b}\lambda_{\mathbf b} f(\mathbf w(\mathbf b))$, where $f$ is a polynomial in $\rho$
variables of degree at most $R+2$. This sum is zero by~\eqref{eq:all-zero}.
All the blocks at this sample are zero, contradicting Step 1.
The rank lower bound is proved.

\paragraph{Extracting a Satisfying Assignment.}
We can recover a satisfying assignment directly from a low-rank matrix in the space $\mathcal{S}$. Given a nonzero matrix in $\Smat$ of rank $\rho\le R$, choose any basis
from its actual columns and compute the polynomials $L_i$ as in Step 3.
These polynomials depend only on that column basis. Enumerate the
$2^\rho$ tuples $\mathbf e$ and
test whether each candidate assignment $(L_1(\mathbf e),\ldots,L_N(\mathbf e))$ is Boolean and
satisfies the input equations. At least one must pass: if all candidates
failed, the selector argument would again force the nonzero block to be
zero. Thus, we obtain a short list containing a satisfying assignment. The sample and the
weighted sums prove that the list works; constructing the list itself
requires only the column basis.

\paragraph{Relationship to Moment Reconstruction.} Our selector representation in Step 5 is an instance of the Boolean moment decomposition described by Laurent~\cite{Laurent2003}: a linear functional, meaning a linear map from polynomials to field elements, is written as a weighted sum of evaluations at Boolean points. It is also related to finite-rank moment reconstruction~\cite{LaurentMourrain2009,Mourrain2018}, where column relations allow one to recover a small set of points and their weights. Here Boolean selectors give the representation directly from the Boolean
identities and the degree bounds checked above. The weights may be arbitrary field elements.

\paragraph{Relationship to Proof Systems.} Multiplying the input equations by bounded-degree
polynomials is the basic operation in Nullstellensatz and polynomial
calculus proofs~\cite{BeameImpagliazzoKrajicekPitassiPudlak1996,
Buss1998Designs,CleggEdmondsImpagliazzo1996}. A Nullstellensatz refutation writes the constant $1$ as a polynomial combination of the input equations and the Boolean equations. Polynomial calculus derives equations by linear combinations and multiplication by variables.
Here the weighted sums play the dual role: they vanish on the allowed
multiples of the equations. Low rank reduces the number of coordinates,
so Boolean selectors of manageable degree suffice to force all these
sums to vanish.

\subsection{Computing the Space Without Listing All Assignments}\label{sec:algorithm}

The definition of $\Vmat$ used $2^N$ matrices $\mathbf A(\mathbf b)$. We now give an iterative algorithm to find this space that keeps only a basis at each stage.

During this algorithm, temporarily set the bits not yet assigned to
zero. Suppose the next bit is $b_i$, and we change it from zero to
$\varepsilon\in\bits$. At each sample, this changes a weight form by
\[
 \ell_{j,t}\longmapsto\ell_{j,t}
                       +\varepsilon(2^{j-1}t)^i.
\]
For example, a square changes by
$\ell^2\mapsto\ell^2+2c\ell+c^2$ for the appropriate constant $c$.
In general, the binomial formula writes every updated weight of degree
at most $R$ as a linear combination of the old weights of degree at most
$R$. This is why the list includes \emph{all} lower degrees.

In the left vector in~\eqref{eq:assignment-stack}, the coordinates are
$h(\mathbf b)b_k$ for all weights $h$ and $0\le k\le N$.
For $k\ne i$, the bit $b_k$ does not change, so the updated coordinate
is the same linear combination of the old coordinates $h(\mathbf b)b_k$.
For $k=i$, its new value is $\varepsilon$ times the updated weight,
which is a linear combination of the coordinates $h(\mathbf b)b_0=h(\mathbf b)$.
Thus there is an explicit linear map $\mathbf U_{i,\varepsilon}$ on the left
vector that performs this update for every current partial assignment.
There is also a linear map $\mathbf V_{i,\varepsilon}$ on the right vector:
it leaves other coordinates alone and replaces coordinate $i$ by
$\varepsilon$ times coordinate zero.
The whole matrix therefore updates by the linear map
\[
 \mathbf A\longmapsto \mathbf U_{i,\varepsilon} \mathbf A \mathbf V_{i,\varepsilon}^T.
\]

Start from the one-dimensional span of the encoding with all bits zero.
At step $i$, apply both maps, for $\varepsilon=0$ and $1$, to a basis of
the current space. Keep a basis of the span of the resulting matrices,
using Gaussian elimination on their lists of entries.
Inductively this is exactly the span of all encodings of the first $i$
bits, with the others temporarily zero: linear maps take a spanning set
to a spanning set of its image. After $N$ steps the space is exactly
$\Vmat$.

To make the output deterministic and ordered, fix the numerical order
of samples, lexicographic order of exponent tuples, and row-by-row order
of matrix entries. In every Gaussian elimination use the first available
nonzero pivot in these orders and retain the resulting basis in pivot order.

The basis size never exceeds the number $m(N+1)$ of entries in a
rectangular matrix. Every transition and every elimination therefore
takes polynomially many field operations in $N$ and $m$. Intersecting
with the constraints in~\eqref{eq:space} is another homogeneous linear
system, giving an ordered basis of $\Smat$. A satisfying assignment's matrix is
explicit, so its coordinates in that basis are found by one further
linear solve. Appending zero columns requires no additional argument.
All operations have bit complexity polynomial also in $N'$ and $\log p$.

Finally,
\[
 m=(N+1)(2NR+1)\binom{2R}{R}=O(N^2R4^R).
\]
For $R=\lfloor\log N\rfloor$ we have $4^R\le N^2$, hence
$m=O(N^4\log N)$. It follows that $\log m=O(\log N)$, while the
NO-case rank bound $R+1$ is greater than $\log N = \Omega(\log m)$, as claimed in
Proposition~\ref{prop:construction}. Since $m>N$ and $m>2NR$,
every prime $p\ge2^m$ satisfies~\eqref{eq:prime}. This proves
Theorem~\ref{thm:warmup}.

\begin{remark}
The algorithm uses the same linear-algebra principle
that underlies deterministic identity testing for algebraic branching
programs~\cite{RazShpilka2005}: after each step, retain a basis of the
reachable vectors instead of enumerating all paths. In our setting the
paths are partial Boolean assignments, and assigning the next bit acts
linearly on their matrix encodings. The same basis computation appears in the earlier algorithm of Cardon and Crochemore~\cite{CardonCrochemore1980} for reducing the dimension of a system of linear transitions.
\end{remark}

\section{A Larger Gap for Bounded-Factor MinRank}\label{sec:polylogminrank}

We now prove the supplied-prime version of the MinRank hardness result
from the introduction.

\begin{theorem}[A Polylogarithmic Gap]\label{thm:minrank}
For every constant $c>0$, there is a deterministic algorithm that,
given a SAT instance $\varphi$, first computes a matrix order
$m=\poly(|\varphi|)$. Given any prime $p\ge 2^m$, it then computes
an ordered basis $\mathbf M_1,\ldots,\mathbf M_k$ of a space
$\Smat\subseteq\Fp^{m\times m}$ such that:
\begin{enumerate}
\item If $\varphi$ is satisfiable, every satisfying witness efficiently
yields $\mathbf a\in\Fp^k$ and $\mathbf u,\mathbf v\in\Fp^m$ with
\[
 \sum_{j=1}^k a_j\mathbf M_j=\mathbf u\mathbf v^T\ne0,
 \qquad \mathbf v\in\bits^m.
\]
\item If $\varphi$ is unsatisfiable, every nonzero matrix in $\Smat$
has rank at least $(\log m)^c$.
\end{enumerate}
The reduction and witness conversion run in time polynomial in
$|\varphi|$ and $\log p$, with an exponent that may depend on $c$.
The order $m$ is independent of the supplied prime.
\end{theorem}

The warm-up used bit pairs in each assignment table and weights of
degree $R$ to rule out rank at most $R$. Its number of weights grows
exponentially in $R$, so this gives a polynomial-size construction
only when $R=O(\log N)$. We will keep more bit products in each
assignment table. These extra entries, together with the corresponding
multiplied equations, let a small weight degree force a much larger
rank. Since the number of extra bit factors is constant, enlarging
the table costs only a polynomial factor.

\subsection{Keeping More Products in Each Assignment Table}

As in Section~\ref{sec:tables}, encode $\varphi$ by quadratic equations
$q_1,\ldots,q_{N'}$ on $\mathbf b\in\bits^N$. Fix an integer $d\ge1$;
its value will depend only on the desired constant $c$.
For $S\subseteq[N]$, write $b_S=\prod_{i\in S}b_i$, with
$b_\varnothing=1$. Replace the vector from the warm-up by
\[
 \mathbf v(\mathbf b)=(b_S)_{S\subseteq[N],\ |S|\le d}.
\]
The table $\mathbf v(\mathbf b)\mathbf v(\mathbf b)^T$ now records
products of up to $2d$ bits. It still has rank one and a Boolean right
factor. For fixed $d$, it has only polynomially many entries.

For a list of weights $\mathcal H$, use the stack
in~\eqref{eq:assignment-stack} with this larger $\mathbf v$ and set
$\Vmat=\Span\{\mathbf A(\mathbf b):\mathbf b\in\bits^N\}$.
Thus an entry of
$\mathbf A=\sum_{\mathbf b}\lambda_{\mathbf b}\mathbf A(\mathbf b)$ is
\[
 \mathbf A_{(h,S),T}
 =\sum_{\mathbf b}\lambda_{\mathbf b}h(\mathbf b)b_Sb_T,
 \qquad |S|,|T|\le d.
\]
The extra entries also allow us to test an equation after multiplying
it by up to $2d-2$ bits. Define $\Smat$ by imposing
\begin{equation}\label{eq:polylog-tests}
 \sum_{\mathbf b}\lambda_{\mathbf b}
 h(\mathbf b)q_s(\mathbf b)b_U=0
 \qquad(h\in\mathcal H,\ 1\le s\le N',\ |U|\le2d-2).
\end{equation}
These are linear tests on matrix entries. Indeed, every monomial of
$q_sb_U$ has degree at most $2d$; after using $b_i^2=b_i$, it splits
as $b_Sb_T$ with $|S|,|T|\le d$. The tests are consequently independent
of the chosen representation of $\mathbf A$. A satisfying assignment
passes every test and gives a nonzero rank-one matrix whenever the
constant weight is included.

The use of higher-degree Boolean assignment monomials
in a rank-one moment matrix already appears in the arbitrary-finite-field
MinRank reduction of Guruswami, Ren, and Tang~\cite{GuruswamiRenTang2026}.
Here we keep the assignment degree $d$ fixed and add the weighted
blocks from the warm-up. The purpose of the weights is to obtain a
growing rank gap without increasing the assignment degree with the
input size.

\subsection{A Rank Bound from One-Bit Selectors}

We first explain why the larger table helps. A full selector in the
warm-up fixes every replacement coordinate. We will instead use one
bit at a time. The two one-bit selectors $x_i$ and $1-x_i$ will give
two smaller tables whose ranks give an additive lower bound. Repeating this argument gives
the binomial coefficient in the following lemma.

Recall that a linear functional $L$ on polynomials assigns a field value to each polynomial and respects addition and scalar multiplication. It is determined by its values on monomials. The weighted sums above are examples. Given row degree $a$ and column degree $b$, form the table
\[
 \mathbf H_{a,b}(L)=\bigl(L(uv)\bigr)_{\deg u\le a,\ \deg v\le b},
\]
whose indices are ordinary monomials, including $1$.
We require the equations and Boolean identities to hold inside $L$
after multiplication by any polynomial of the allowed degree.

These multiplied equations are the familiar
constraints underlying bounded-degree Nullstellensatz and polynomial
calculus arguments~\cite{BeameImpagliazzoKrajicekPitassiPudlak1996,CleggEdmondsImpagliazzo1996,Buss1998Designs}.
Here we use them directly to prove a rank bound. We do not assume
that $L$ is a probability distribution or invoke a proof-complexity
lower bound.

\begin{lemma}[Rectangular Boolean Tables]\label{lem:pascal}
Suppose the quadratic polynomials $q_1,\ldots,q_{N'}$ have no common
Boolean zero. Let $a,b$ be
nonnegative integers, and let $L$ be a linear functional on polynomials
in $x_1,\ldots,x_N$ of degree at most $a+b$. Suppose
\[
 L((x_i^2-x_i)f)=L(q_sf)=0
 \qquad(\deg f\le a+b-2)
\]
for every $i,s$ and every such $f$. If $\mathbf H_{a,b}(L)$ is
nonzero, then
\[
 \rank\mathbf H_{a,b}(L)\ge\binom{a+b}{a}.
\]
When $a=0$ or $b=0$, the conclusion requires neither kind of identity.
\end{lemma}
\begin{proof}
Induct on $a+b$. If either degree is zero, nonzero rank is at least
one, as required. Suppose both are positive and write $e=a+b$.

Some variable $x=x_i$ makes both functionals
\[
 L_x(f)=L(xf),\qquad L_{1-x}(f)=L((1-x)f)
\]
nonzero on polynomials of degree at most $e-1$. To prove this when
$L(1)=0$, take a monomial $u$ of least degree with $L(u)\ne0$.
Remove one factor $x$. Minimality gives $L(u/x)=0$, so the two
functionals take the nonzero values $L(u)$ and $-L(u)$ on $u/x$.

If $L(1)\ne0$ and no variable has two nonzero functionals, then for
each $i$ one of the two is identically zero. There is therefore a bit
$\epsilon_i$ such that
\[
 L(x_if)=\epsilon_iL(f)\qquad(\deg f\le e-1).
\]
Writing $\boldsymbol\epsilon=(\epsilon_1,\ldots,\epsilon_N)^T$ and
removing one factor at a time yields
$L(u)=L(1)u(\boldsymbol\epsilon)$ for every monomial of degree at
most $e$. In particular, $L(q_s)=0$ implies
$q_s(\boldsymbol\epsilon)=0$ for every $s$, contradicting
unsatisfiability. This proves the claim about $x$.

In $\mathbf H=\mathbf H_{a,b}(L)$, take the row combinations
indexed by $xu$, $\deg u\le a-1$, and the column combinations
indexed by $(1-x)v$, $\deg v\le b-1$. Their intersection is zero:
\[
 L(x(1-x)uv)=0.
\]
Indeed, $x(1-x)=-(x^2-x)$ and $\deg(uv)\le a+b-2$, so this
is one of the assumed Boolean identities.
Write the two combination maps as $\mathbf P,\mathbf Q$, so that
$\mathbf P\mathbf H\mathbf Q=0$. Restrict $\mathbf P$ to the image
of $\mathbf H$. The image of this restricted map is the image of
$\mathbf P\mathbf H$, so rank--nullity gives
\[
 \rank\mathbf H
 =\rank(\mathbf P\mathbf H)
   +\dim(\operatorname{im}\mathbf H\cap\ker\mathbf P).
\]
Every column of $\mathbf H\mathbf Q$ belongs to the image of
$\mathbf H$, and $\mathbf P\mathbf H\mathbf Q=0$ says that it also
belongs to the kernel of $\mathbf P$. Thus the intersection above
has dimension at least $\rank(\mathbf H\mathbf Q)$, giving
\[
 \rank\mathbf H\ge
 \rank(\mathbf P\mathbf H)+\rank(\mathbf H\mathbf Q).
\]
The two terms are the ranks of $\mathbf H_{a-1,b}(L_x)$ and
$\mathbf H_{a,b-1}(L_{1-x})$. Their row-column products include every
monomial of degree at most $e-1$, so the nonvanishing of the two
functionals implies that both tables are nonzero.
Both also inherit the required identities. A test for either smaller
table has multiplier degree at most $e-3$. Multiplying this multiplier
by $x$ or $1-x$, respectively, raises its degree to at most $e-2$,
which is allowed in the identities for $L$. If one of the smaller
degrees is zero, no identities are needed. Induction now gives
\[
 \rank\mathbf H\ge
 \binom{a+b-1}{a-1}+\binom{a+b-1}{a}
 =\binom{a+b}{a}.
\]
\end{proof}

The splitting into $L_x$ and $L_{1-x}$ is the algebraic counterpart of conditioning on a bit in Boolean moment relaxations~\cite{Laurent2003}. Here it yields an additive rank bound without requiring nonnegative weights.

With row degree $r+d$ and column degree $d$, the lemma gives rank
at least $\binom{r+2d}{d}$, which grows as $r^d$ for fixed $d$.
Thus even $r$ of order $\sqrt{\log N}$ can give any fixed power of
a logarithm by choosing $d$ large enough. We cannot explicitly store
all original-bit monomials of degree $r+d$. As in the warm-up, the
weights will use a small list of linear forms instead.

\subsection{Choosing the Weights}

Set
\begin{equation}\label{eq:polylog-parameters}
 r=\max\{1,\lfloor\sqrt{\log(N+2)}\rfloor\},
 \qquad R=\binom{r+2d}{d}.
\end{equation}
Here $r$ is the weight degree and $R$ is the rank we aim to force.
We will need column relations to remain valid after multiplication by
up to $d-1$ additional bits. The soundness proof shows that preserving
these extra relations requires spanning a space of dimension at most
$R^2$, which explains the number of forms below.

We reuse the forms obtained by evaluating at $t,2t,4t,\ldots$ from
Section~\ref{sec:weights}, and thus the same rank-condensing idea from~\cite{ForbesShpilka2012,GuruswamiKopparty2016,ForbesSaptharishiShpilka2014,ForbesGuruswami2015}.
The change is the space to which the construction is applied: it must
now preserve relations among collections of columns, rather than only
among the bit columns.

Use the same geometric forms as in~\eqref{eq:forms}, now with
$1\le j\le R^2$ and $0\le t\le N(R^2+r)$:
\[
 \ell_{j,t}(\mathbf b)=\sum_{i=0}^N(2^{j-1}t)^i b_i,
 \qquad b_0=1.
\]
At each $t$, include all weights
\begin{equation}\label{eq:polylog-weights}
 h_{\boldsymbol\alpha,t}(\mathbf b)
 =\prod_{j=1}^{R^2}\ell_{j,t}(\mathbf b)^{\alpha_j},
 \qquad \alpha_j\in\mathbb Z_{\ge0},\quad \sum_j\alpha_j\le r.
\end{equation}
Retain separate blocks for different labels even if their weights
coincide. The constant weight is present. Linear combinations of
the blocks at $t$ allow every polynomial of degree at most $r$ in
these forms as a weight.

The construction works whenever
\begin{equation}\label{eq:polylog-prime}
 p>\max\{2^N,N(R^2+r)\}.
\end{equation}
The first bound lets us apply Lemma~\ref{lem:spanning}; the second
keeps the sample values distinct. We will check below that the prime
bound in Theorem~\ref{thm:minrank} implies both.

\subsection{A Nonzero Low-Rank Matrix Gives a Contradiction}

Suppose the equations are unsatisfiable and a nonzero
$\mathbf A\in\Smat$ has rank $\rho<R$. Fix a representation
$\mathbf A=\sum_{\mathbf b}\lambda_{\mathbf b}\mathbf A(\mathbf b)$.
We will construct a nonzero table $\mathbf H_{r+d,d}(L)$ of rank at
most $\rho$ to which Lemma~\ref{lem:pascal} applies.
As in the warm-up, the first step is to express the original bits
using a small number of geometric forms inside weighted sums. We
then transfer the multiplied equations to these replacement
coordinates and compare the resulting table with the original stack.

\paragraph{Column Relations That Survive Additional Multiplication.}
Let $\mathbf C_T$ denote column $T$ of the whole stack. For each
$0\le i\le N$, stack the columns
\[
 \mathbf c_i=(\mathbf C_{iU})_{|U|\le d-1},
 \qquad iU=U\cup\{i\}\ (i\ge1),\quad 0U=U.
\]
A relation among these vectors holds after multiplying a bit by
every monomial of degree at most $d-1$. This is the additional
freedom that we need to use the larger table. Their span has dimension
\begin{equation}\label{eq:polylog-bundles}
 1\le\delta\le\rho^2.
\end{equation}
For the upper bound, choose $\rho$ actual rows $(h_v,S_v)$ such that
restricting to those rows is injective on the column space. Apply this
restriction separately to every column in each $\mathbf c_i$.
It remains injective on the span of the $\mathbf c_i$: if all the
restricted columns of a linear combination vanish, injectivity on
the column space makes each full column vanish.
If $S_v\ne\varnothing$, write $S_v=\{a_v\}\cup U_v$ with
$|U_v|\le d-1$; otherwise take $a_v=0$ and $U_v=\varnothing$.
The symmetry of bit products gives
\[
 \mathbf A_{(h_v,S_v),iU}
 =\sum_{\mathbf b}\lambda_{\mathbf b}
    h_v(\mathbf b)b_{a_v}b_{U_v}b_i b_U
 =\mathbf A_{(h_v,iU_v),a_vU}.
\]
For each $v$ and $U$, the column index $a_vU$ on the right is
independent of $i$. Thus all restricted coordinates of $\mathbf c_i$
are obtained by taking fixed coordinates from the tuple of rows
$(\mathbf A_{(h_v,iU_v),*})_{v=1}^{\rho}$.
Each component belongs to the $\rho$-dimensional row space of
$\mathbf A$. The space of all such $\rho$-tuples therefore has
dimension $\rho\cdot\rho$, so its image under this fixed coordinate
map has dimension at most $\rho^2$. Since the restriction is
injective, the original vectors $\mathbf c_i$ span at most
$\rho^2$ dimensions as well.
The lower bound follows because every column occurs in one of the
vectors $\mathbf c_i$ and $\mathbf A\ne0$.

If $\sum_i\kappa_i\mathbf c_i=0$, put
$K(\mathbf b)=\sum_i\kappa_i b_i$. Reading this relation in every
row gives
\begin{equation}\label{eq:polylog-replace}
 \sum_{\mathbf b}\lambda_{\mathbf b}K(\mathbf b)h(\mathbf b)f(\mathbf b)=0
 \qquad(h\in\mathcal H,\ \deg f\le2d-1).
\end{equation}
Indeed, the rows and the coordinates indexed by $U$ give the identity
for $f=b_Sb_U$, with $|S|\le d$ and $|U|\le d-1$. Every monomial
of degree at most $2d-1$ reduces on Boolean assignments to such a
product. Linearity also allows arbitrary linear combinations of
the weights.

\paragraph{Choosing One Useful Sample.}
Lemma~\ref{lem:spanning} applied to the $\mathbf c_i$ gives a nonzero
polynomial $\Delta(T)$ of degree at most $N\delta$. Its nonvanishing
ensures that the first $\delta$ geometric combinations form a basis.
As in the warm-up, this sample must also retain a nonzero entry.
Choose a nonzero entry with labels $(\boldsymbol\alpha,t_0,S),T_0$.
Replacing its sample value by an indeterminate gives
\[
 F(T)=\sum_{\mathbf b}\lambda_{\mathbf b}b_Sb_{T_0}
       \prod_{j=1}^{R^2}\ell_{j,T}(\mathbf b)^{\alpha_j}.
\]
This polynomial is nonzero and has degree at most $Nr$. Since
$\delta\le\rho^2<R^2$, the product $F\Delta$ has degree at most
$N(R^2+r)$. Some sample in our list therefore has
$F(t)\Delta(t)\ne0$.

Fix this $t$, abbreviate $\ell_j=\ell_{j,t}$, and let $\mathbf A_t$
be the stack of all blocks at this sample. Then $\mathbf A_t\ne0$.
Expressing each $\mathbf c_i$, $i\ge1$, in the basis of geometric
combinations gives coefficients $\mu_{ij}$ with
\[
 \mathbf c_i=\sum_{j=1}^{\delta}\mu_{ij}
                  \sum_{k=0}^N(2^{j-1}t)^k\mathbf c_k,
 \qquad s_i(\mathbf b)=\sum_{j=1}^{\delta}\mu_{ij}\ell_j(\mathbf b).
\]
The coefficients of $b_i-s_i(\mathbf b)$ give a relation among the
$\mathbf c_k$. We may therefore use this difference as $K$ in
\eqref{eq:polylog-replace}. As before, this is a replacement inside
weighted sums; the individual values $s_i(\mathbf b)$ need not be
Boolean. The literal constant $1$ is left unchanged.

\paragraph{Transferring the Equations to the Replacement Coordinates.}
We first record the degree bounds. Let $P$ be a product of at
most $r+2d-1$ affine linear factors, at most $2d-1$ of which are
arbitrary forms in the original bits; require every other factor
to be affine linear in the $\ell_j$. Then
\begin{equation}\label{eq:polylog-budget}
 \sum_{\mathbf b}\lambda_{\mathbf b}(b_i-s_i(\mathbf b))P(\mathbf b)=0.
\end{equation}
To see this, put enough of the latter factors into the weight to
leave at most $2d-1$ factors. There are at most $2d-1$ arbitrary
original-bit factors, so the factors we remove can all be chosen
from the forms in the $\ell_j$. At most $r$ factors are put into the weight.
Their product is an allowed weight combination, and expanding the
remaining factors lets us apply~\eqref{eq:polylog-replace}.

Define $L$ on polynomials of degree at most $r+2d$ by
\[
 L(f)=\sum_{\mathbf b}\lambda_{\mathbf b}
                 f(s_1(\mathbf b),\ldots,s_N(\mathbf b)).
\]
We claim that, for every $\deg f\le r+2d-2$,
\begin{equation}\label{eq:polylog-moment-tests}
 L((x_i^2-x_i)f)=L(q_sf)=0.
\end{equation}
For Booleanity, use
\[
 s_i^2-s_i=(s_i-b_i)(s_i+b_i-1)+(b_i^2-b_i).
\]
After multiplication by $f(s_1,\ldots,s_N)$, consider each monomial
of $f$ separately. Apart from the factor $s_i-b_i$, the first term
has at most $(r+2d-2)+1=r+2d-1$ factors. All factors coming from
$f$ are linear in the $\ell_j$; only $s_i+b_i-1$ needs to be treated
as an arbitrary affine form in the original bits. Since $d\ge1$,
this is within the allowance of $2d-1$ such factors.
Equation~\eqref{eq:polylog-budget} therefore makes its weighted sum
zero. The second term vanishes on every original assignment.

For the source equations, first consider
$q_s(\mathbf b)f(s_1,\ldots,s_N)$. In each monomial of the latter
factor, put enough factors into the weight to leave at most $2d-2$.
At most $r$ factors are needed, and their product expands in the
available weights because each $s_i$ is linear in the $\ell_j$.
Expand the remaining factors in the original bits and use $b_i^2=b_i$.
Every resulting monomial is some $b_U$ with $|U|\le2d-2$, so
\eqref{eq:polylog-tests} gives
\[
 \sum_{\mathbf b}\lambda_{\mathbf b}
 q_s(\mathbf b)f(s_1(\mathbf b),\ldots,s_N(\mathbf b))=0.
\]
Now replace the one or two bit factors in each monomial of $q_s$
by the corresponding $s_i$. For a quadratic monomial, write
\[
 b_i b_j-s_i s_j=(b_i-s_i)b_j+s_i(b_j-s_j).
\]
After multiplying by any monomial of $f(s_1,\ldots,s_N)$, each term
has at most $r+2d-1$ factors apart from its replacement difference.
At most one of these, namely $b_j$ in the first term, is an
arbitrary original-bit factor; all the others are linear in the
$\ell_j$. Thus~\eqref{eq:polylog-budget} justifies both replacements.
Linear terms use the same rule with one fewer factor, and constants
stay unchanged. This proves the claim.

\paragraph{Comparing the Ranks of the Two Tables.}
Set $\mathbf H=\mathbf H_{r+d,d}(L)$. This table is obtained by row
and column combinations of $\mathbf A_t$. For a row monomial
evaluated at $(s_1,\ldots,s_N)$, put at most $r$ factors into the
weight and at most $d$ into the row. Each $s_i$ is linear in the
$\ell_j$, so the first group expands in the available weights.
The second group expands in monomials of the original bits of
degree at most $d$; using $b_i^2=b_i$ makes these available row
monomials. A column monomial has at most $d$ factors and similarly
expands in the columns of $\mathbf A_t$.
The coefficients in the row expansion depend only on the row
monomial, and those in the column expansion depend only on the
column monomial. Substituting both expansions into $L(uv)$ therefore
gives fixed row and column operations that produce all entries of
$\mathbf H$ simultaneously. These operations cannot increase rank,
and $\mathbf A_t$ is a submatrix of $\mathbf A$, so
\[
 \rank\mathbf H\le\rank\mathbf A_t\le\rho.
\]

We must also show $\mathbf H\ne0$. Take an entry
$\sum_{\mathbf b}\lambda_{\mathbf b}h_{\boldsymbol\alpha,t}(\mathbf b)b_Sb_T$
of $\mathbf A_t$. Replace its at most $2d$ bit factors by $s_i$,
one at a time. At each replacement, the other factors number at
most $r+2d-1$. At most $2d-1$ are still original-bit factors;
the weight factors and the already replaced factors are all linear
in the $\ell_j$. Thus~\eqref{eq:polylog-budget} preserves the weighted
sum at every step. Then replace each weight factor
\[
 \ell_j(\mathbf b)=1+\sum_{i=1}^N(2^{j-1}t)^i b_i
 \quad\hbox{by}\quad
 1+\sum_{i=1}^N(2^{j-1}t)^i s_i(\mathbf b).
\]
The difference is a linear combination of the $b_i-s_i$; all other
factors are now affine linear in the $\ell_j$, and there are at
most $r+2d-1$ of them. Thus~\eqref{eq:polylog-budget} applies again.
The resulting entry is
\[
 L\!\left(
   \left[\prod_{j=1}^{R^2}
       \left(1+\sum_{i=1}^N(2^{j-1}t)^i x_i\right)^{\alpha_j}\right]
   x_Sx_T\right).
\]
Here $x_S=\prod_{i\in S}x_i$, as for the bit monomials.
The bracketed factor times $x_S$ has degree at most $r+d$, and
$x_T$ has degree at most $d$. Hence $\mathbf A_t$ is also obtained
by row and column combinations of $\mathbf H$. Since
$\mathbf A_t\ne0$, we have $\mathbf H\ne0$.
Lemma~\ref{lem:pascal} and~\eqref{eq:polylog-moment-tests} now give
\[
 \rho\ge\rank\mathbf H\ge\binom{r+2d}{d}=R,
\]
a contradiction. The large table $\mathbf H$ was needed only in
the proof; the reduction never constructs it.

\subsection{Computing the Space and Choosing the Gap}

The number of columns is $\sum_{j=0}^{\min(d,N)}\binom Nj$.
There are $\binom{R^2+r}{r}$ weights at each sample, so the row count is
\begin{equation}\label{eq:polylog-size}
 m=\left(\sum_{j=0}^{\min(d,N)}\binom Nj\right)
       [N(R^2+r)+1]\binom{R^2+r}{r}.
\end{equation}
Append zero columns to make the matrices square. This preserves
their ranks and extends the right factor by zeros, so it remains
Boolean. Also $m>N$ and $m>N(R^2+r)$. Every supplied prime
$p\ge2^m$ therefore satisfies~\eqref{eq:polylog-prime}.

The basis algorithm from Section~\ref{sec:algorithm} still applies.
Keep unassigned bits zero and set the next bit $b_i$ to
$\varepsilon\in\bits$. Each weight form is translated by a known
constant, so the binomial formula expresses its updated monomials
using weights of degree at most $r$. A bit monomial containing $i$
becomes $\varepsilon$ times the monomial with $i$ removed; all other
bit monomials stay unchanged. Consequently, setting the bit is a
linear map on the left vector $(h(\mathbf b)b_S)_{h,S}$ and a
linear map on the right vector $(b_T)_T$.

Start with the all-zero assignment matrix. At each bit, apply both
updates to a basis of the current span and retain a basis of their
combined span. After $N$ steps this computes $\Vmat$ exactly.
Intersect it with the linear tests~\eqref{eq:polylog-tests} to obtain
$\Smat$. Fix index and pivot orders to obtain a deterministic
ordered basis. A satisfying assignment explicitly gives both
rank-one factors; a linear solve gives its coordinates in that
basis. The basis size is at most the number of matrix entries, so
these operations take time polynomial in $N,N',m,\log p$.

It remains to choose $d$. For fixed $d$, the parameters
in~\eqref{eq:polylog-parameters} satisfy
\[
 R=\Theta_d(r^d),\qquad
 \log\binom{R^2+r}{r}
 \le r\log(R^2+r)
 =O_d(\sqrt{\log N}\log\log N)=o(\log N).
\]
Together with~\eqref{eq:polylog-size}, this gives
$N\le m=N^{O_d(1)}$ and
\[
 R=\Theta_d((\log N)^{d/2})
  =\Theta_d((\log m)^{d/2}).
\]
Choose an integer $d>2c$. For all sufficiently large $N$, this
implies the exact bound $R\ge(\log m)^c$. Pad the quadratic encoding
with unused bits to exceed this fixed threshold; the padding depends
only on $c$. The construction is then polynomial time, its NO-case
rank bound is the one in Theorem~\ref{thm:minrank}, and its YES-case
factor and witness conversion were established above. Finally,
\eqref{eq:polylog-size} depends only on $N$ and $d$, so $m$ can be
computed before the prime is supplied. This proves the theorem.

\section{Putting It All Together}\label{sec:together}
We now prove Theorems~\ref{thm:main} and~\ref{thm:minrank-informal},
followed by the cryptographic applications in
Section~\ref{sec:applications}. Both reductions compute the matrix
order $m$ before receiving the prime.

\paragraph{Choosing the Prime.}
Sample independent uniform integers from $[2^m,2^{m+1})$ and apply a
deterministic polynomial-time primality test~\cite{AgrawalKayalSaxena2004}.
A proportion $\Omega(1/m)$ of this interval is
prime~\cite[Corollary~3]{RosserSchoenfeld1962Primes}, so
$\lceil Cm(n+2)\rceil$ trials fail with probability at most
$2^{-n-2}$ for a sufficiently large absolute constant $C$.
This takes polynomial time. For witness encryption, sampling is
independent of the encrypted bit, and failure returns the common abort
symbol from Section~\ref{sec:we}.

\paragraph{Witness Encryption.}
Fix $c>0$ and apply Theorem~\ref{thm:minrank} with exponent $c+1$.
Pad the Boolean quadratic encoding with unused variables to make
$N\ge16n$. The resulting matrix order satisfies $m\ge N$ and
$m=\poly(n)$. Choose $p$ as above and use the encryption algorithm
of Section~\ref{sec:we}. For each supplied prime, the reduction is
deterministic, preserves witnesses, and gives NO-case rank at least
$(\log m)^{c+1}$. Thus we may apply
Theorem~\ref{thm:we-from-minrank} with
\[
 R=(\log n)^{c+1}.
\]
We have $R\to\infty$ and $R\log n=o(m)$. For sufficiently large $n$,
$R/8\ge(\log n)^c$, giving advantage at most $n^{-(\log n)^c}$
against every adversary making at most $n^{(\log n)^c}$ queries.
This holds for every sampled prime. Since both messages use the same
prime distribution and abort symbol, averaging preserves the bound
without an additive error.

By Lemma~\ref{lem:sampling} and Proposition~\ref{prop:correctness},
including prime-sampling failure gives correctness error at most
\[
 2^{-n-2}+\delta+m^2 2^{-m}
 \le 2^{-n-2}+2^{-n-2}+2^{-n-1}=2^{-n}.
\]
All inequalities hold above a fixed threshold depending only on $c$
and the NP relation. Pad smaller inputs and their security parameters
to this threshold; increasing the internal size improves the required
bounds and preserves polynomial running time.
For a general NP relation, compile its verifier into a Boolean circuit
and extend each valid witness by its gate values. Choose $N$ at least
sixteen times both the circuit size and the original $n$. Then
$m\ge16n$ is still polynomial in the original $n$, so the same bounds
prove Theorem~\ref{thm:main} with its stated size convention.

\paragraph{Randomized MinRank Hardness.}
Apply Theorem~\ref{thm:minrank} with exponent $c$ and the prime sampler
with $n=\max\{2,|\varphi|\}$. On success, the output satisfies the
promises of Theorem~\ref{thm:minrank-informal}, with
$p\in[2^m,2^{m+1})$. On failure, output the zero matrix space of order
four over $\mathbb F_{17}$. This is a NO instance for every $c$, and
$2^4\le17<2^5$. Thus NO inputs always map to NO instances, while YES
inputs map to YES instances except with probability at most $2^{-n-2}$.
This proves the randomized polynomial-time hardness result. Dropping
the Boolean-factor condition gives the same gap for homogeneous MinRank.

\paragraph{The Cryptographic Applications.}
We now prove Corollaries~\ref{cor:pke}--\ref{cor:concrete-applications}.
To encrypt polynomially many bits, apply Theorem~\ref{thm:main}
independently to each bit with a fresh group context. A union bound
and a bitwise hybrid give negligible correctness error and security
on false statements. The first three corollaries use ordinary, group-free auxiliary
primitives. In hybrids based on their computational assumptions, the
reduction knows the witness-encryption coins and simulates the
independent group contexts by lazily assigning random encodings.
In a witness-encryption hybrid, it forwards the challenge context and
simulates the others. The statement and adversary's state are fixed
before the fresh challenge context is sampled, so the pointwise
guarantee of Theorem~\ref{thm:main} applies.

\begin{proof}[Proof of Corollary~\ref{cor:pke}]
Use the construction of Garg, Gentry, Sahai, and Waters~\cite{GargGentrySahaiWaters2013}.
One-way functions give a pseudorandom generator~\cite{HILL1999PRG}
$G:\bits^\lambda\to\bits^{2\lambda}$. Sample $s\leftarrow\bits^\lambda$,
publish $y=G(s)$, and retain $s$ as the secret key. Encrypt under the
NP statement that $y$ is in the image of $G$. The seed is a valid
witness, giving correctness.

For either challenge bit, pseudorandomness lets us replace $y$ by a
uniform $2\lambda$-bit string. This lies in the image of $G$ with
probability at most $2^{-\lambda}$; otherwise the statement is false,
and Theorem~\ref{thm:main} hides the bit. Replacing the uniform string
by a generator output completes the IND-CPA hybrid. The simulation
above justifies the pseudorandomness steps in the generic-group model.
\end{proof}

\begin{proof}[Proof of Corollary~\ref{cor:abe}]
Apply Waters and Wichs~\cite{WatersWichs2024}, using LWE to instantiate
their one-way functions and statistically sound noninteractive
zero-knowledge proofs~\cite{PeikertShiehian2019NIZK}. Their witness-encryption relation checks a
zero-knowledge proof, a circuit policy, and ordinary auxiliary
conditions; it does not evaluate witness encryption. Thus it remains
an ordinary NP relation, with group queries only in the outer
encryption and decryption algorithms. The simulation above preserves
the security hybrids, giving adaptive security and negligible
correctness error.
\end{proof}

\begin{proof}[Proof of Corollary~\ref{cor:trustless}]
Apply Freitag, Waters, and Wu~\cite{FreitagWatersWu2023}, using their
LWE-based function-binding hashes and a perfectly correct, ordinary
LWE-based public-key encryption scheme~\cite[Section~7.2.1]{CanettiEtAl2018SimplerFiatShamir}.
The latter must be group-free, since its algorithms occur inside the
hash-function circuits and NP relations. Thus we use this scheme,
rather than Corollary~\ref{cor:pke}, for the internal encryption primitive.

The witness-encryption relations check ordinary public-key
encryptions, hash openings, and, for registered ABE, the circuit
policy. Only the outer encryption and decryption algorithms call
witness encryption. Theorem~\ref{thm:main} and the simulation above
therefore preserve semi-static security for flexible broadcast
encryption and policy-selective security for registered ABE, with
the stated restriction on corrupting honestly generated keys.
\end{proof}

\begin{proof}[Proof of Corollary~\ref{cor:concrete-applications}]
Under the assumed instantiation, decryption is an ordinary polynomial-time algorithm and can be represented by a Boolean circuit.

A \emph{compute-and-compare obfuscator} encodes the test $f(x)=y$, preserving its entire Boolean function except with negligible probability. We use its security when $y$ has at least $\lambda$ bits and is pseudorandom given $f$ and the adversary's auxiliary information: the obfuscation can then be simulated using only the circuit size and input/output lengths. Wichs and Zirdelis~\cite[Section~6.2, Theorem~6.4]{WichsZirdelis2017} use this guarantee to convert witness encryption into null obfuscation. We check how negligible witness-encryption error affects their compiler's correctness.

For a Boolean circuit $C$ on $\ell$ input bits, the compiler encrypts
a random string $y$ of polynomial length at least $\lambda$ under the statement that $C$ is
satisfiable, then obfuscates the program accepting exactly when witness
decryption returns $y$. Have decryption reject witnesses not accepted
by $C$. Encrypt each bit with internal security parameter at least
$\lambda+\ell+\lceil\log|y|\rceil$, using public length bounds so
that equal-size circuits use the same parameters. Theorem~\ref{thm:main}
gives error at most $2^{-\lambda-\ell}/|y|$ per bit with polynomial
overhead. A union bound over all $2^\ell$ inputs and all bits shows
that, except with probability $2^{-\lambda}$, the comparison program
agrees with $C$ on every input. The compute-and-compare obfuscator
preserves this entire function with overwhelming probability.
Thus the null obfuscator has global negligible correctness error,
and its security proof is unchanged. If needed, take a majority of
independent copies to make this error exponentially small: a majority
of globally correct copies is correct on every input simultaneously.
Polynomially many copies preserve security by a hybrid. Applying the
construction bitwise covers multiple output bits and a distinguished
output $\bot$.

We recall the properties needed by Nassar, Waters, and Wu~\cite{NassarWatersWu2026}. A \emph{cutoff obfuscator} takes a circuit $C(w,i)$ with $i\in[N]$ and a cutoff $t\in\{0,\ldots,N\}$. It produces a program that evaluates $C(w,i)$ when $i>t$ and outputs $\bot$ otherwise. Moving the cutoff past an index $i$ is computationally hidden when $C(w,i)$ is identically $\bot$ as a function of $w$. The all-off program, corresponding to $t=N$, also hides the underlying circuit.

Their recursive construction uses one block program per interval of indices. An inner block is a program from the next recursive level. An \emph{evaluation hint} supplies such a block, a proof opening a digest to that block, and a further hint for evaluating it. Besides correctness for honestly generated hints, the construction requires \emph{endpoint soundness}: at $t=N$, arbitrary hints must keep the program off; at $t=0$, every accepted hint must give $C(w,i)$~\cite[Definition~4.1]{NassarWatersWu2026}.

Apply the global guarantee to each circuit's full input, including its hint. Conditioned on earlier execution, each
fresh obfuscation has the same error bound for the resulting fixed
circuit. A union bound covers all obfuscations in an execution, even
when earlier outputs determine later circuits. On this good event,
the original correctness and endpoint arguments hold with their
quantifiers over inputs and hints.

In the security hybrids, inner block programs are generated with
pseudorandom coins. Use a pseudorandom function of the form
\[
 F'_{K,r}(j)=F_K(j)\oplus r,
\]
where $r$ is an independent uniform string of the output length.
Replacing $F_K$ by a random function shows that $F'$ is pseudorandom;
the mask also makes each fixed output exactly uniform. The state of a block-generation call consists of its inputs other than its random tape, including the inner circuit and its cutoff. These inputs are independent of $(K,r)$, so we may fix them before sampling $(K,r)$. Thus each block-generation tape has a uniform marginal conditional on that state. The same union bound therefore applies without independence
between blocks. This choice only affects the proof hybrids; the
construction uses fresh randomness.
For an index $j$ and a proposed block $B$, the binding predicate tests whether $B$ equals the output of the block-generation algorithm on the fixed state and the prescribed tape $F'_{K,r}(j)$. On the binding event, every accepted hinted block therefore equals its generated block~\cite[Claims~4.24 and~4.32]{NassarWatersWu2026}. Include the negligible probability of a bad binding key in the good-event bound; a union bound covers the polynomially many keys.

At each null-obfuscation hybrid, the good event makes both challenge circuits identically $\bot$. Pad the two circuits to the same size if needed. Fix the circuits and the auxiliary state chosen before the fresh challenge randomness. For every such fixed choice with two null circuits, nonuniform null-obfuscation security bounds the distinguishing advantage. The bound is uniform over choices of polynomial size: otherwise a nonuniform adversary could hardwire a sequence of violating choices. Averaging therefore preserves the bound. Choices with a non-null challenge have total probability at most the bad-event probability and contribute at most that much to the advantage. This argument does not require the reduction to test whether a circuit is null. This gives the required
cutoff-obfuscation properties with negligible correctness error.
Padding, masking, and repetition cost only a polynomial in the
circuit parameters and $\lambda+\log N$, preserving the recursive
construction's polynomial overhead and $N^\varepsilon$ succinctness
for every fixed $\varepsilon\in(0,1)$.

The indexed positional witness-encryption and broadcast-and-trace
constructions of~\cite{NassarWatersWu2026} now apply. For batch arguments, we need their hash guarantee. Their LWE-based base hash can already be programmed, with a computationally indistinguishable setup, to avoid one efficiently specified output for every input~\cite[Fact~6.3]{NassarWatersWu2026}. The construction extends this guarantee to a list of at most $N$ forbidden outputs per input, where the hash output space has more than $N$ elements. Thus at least one allowed output is available. These input/output pairs form the \emph{programmed relation}. In \emph{avoiding mode}, key generation inserts an obfuscated program that replaces forbidden base-hash outputs with allowed ones. These keys are computationally indistinguishable from ordinary keys. \emph{Statistical correlation intractability} requires that, except with negligible probability over the avoiding key, no input hashes to a forbidden output~\cite[Section~6]{NassarWatersWu2026}.

Use their simpler cutoff obfuscator with independently obfuscated blocks. Except with negligible probability, global correctness makes every hash input avoid the programmed relation in avoiding mode. The hash retains
statistical correlation intractability with negligible error, so
Appendix~A of~\cite{NassarWatersWu2026} gives somewhere-statistically
sound noninteractive batch arguments for NP. Statistical soundness
concerns the reference string programmed at a hidden index, whose
distribution is computationally indistinguishable from ordinary setup.
All computational reductions use only polynomial security of the
assumed instantiation and the stated LWE assumption.
\end{proof}

\clearpage
\bibliographystyle{alpha}
\bibliography{refs}
\end{document}